\documentclass[a4paper,11pt]{article}
\usepackage[english]{babel}
\usepackage{german}
\usepackage{amssymb}
\usepackage{amsfonts}
\usepackage{amsmath}
\usepackage{fullpage}
\usepackage{cite}
\usepackage{vmargin}

\usepackage{todonotes}
\usepackage{subfigure}
\usepackage{comment}

\newtheorem{theorem}{Theorem}

\newtheorem{corollary}{Corollary}

\newtheorem{definition}{Definition}

\newtheorem{lemma}{Lemma}

\newenvironment{proof}[1][Proof]{\noindent\textbf{#1.} }{\ \rule{0.5em}{0.5em}}

\setmarginsrb{2.5cm}{2.25cm}{2.5cm}{3.05cm}{0.3cm}{0.3cm}{-0.3cm}{1.0cm}

\usepackage{fancyhdr}
 { 
}

\usepackage{xspace}

\usepackage{algorithm2e}
\usepackage{algorithmicx}

\newcommand{\NP}{\ensuremath{\mathtt{NP}}\xspace}
\newcommand{\PP}{\ensuremath{\mathtt{P}}\xspace}

\newcommand{\clique}{\textsc{Clique}\xspace}

\newcommand{\tcal}{\ensuremath{\mathcal{T}}\xspace}

\newcommand{\tuple}[1]{\ensuremath{\langle {#1} \rangle}\xspace}

\newcommand{\tmax}{\ensuremath{t_{\max}}\xspace}

\newcommand{\tgraph}{\ensuremath{\tuple{G,\tcal}}\xspace}
\newcommand{\tcalm}{\ensuremath{\mathcal{T}^M_{(\lambda,\mu)}}\xspace}
\newcommand{\tcald}{\ensuremath{\mathcal{T}_{\delta}}\xspace}
\newcommand{\tcaldk}{\ensuremath{\mathcal{T}_{(\delta, \kappa)}}\xspace}
\newcommand{\reach}{\ensuremath{\mathtt{reach}}\xspace}

\newcommand{\merging}[1]{({#1})-\textsc{Merging}\xspace}
\newcommand{\delaying}[1]{{#1}-\textsc{Delaying}\xspace}

\newcommand{\minreach}{\textsc{MinReach}\xspace}
\newcommand{\maxreach}{\textsc{MaxReach}\xspace}
\newcommand{\minavgreach}{\textsc{MinAvgReach}\xspace}
\newcommand{\maxavgreach}{\textsc{MaxAvgReach}\xspace}
\newcommand{\minmaxreach}{\textsc{MinMaxReach}\xspace}
\newcommand{\maxminreach}{\textsc{MaxMinReach}\xspace}
\newcommand{\timed}{\textsc{Timed}\xspace}

\newcommand{\maxsat}{\textsc{Max2Sat(3)}\xspace}

\newcounter{probc}
\newcommand{\probc}[1]{\refstepcounter{probc}\label{#1}}

\usepackage{tikz}
\usetikzlibrary{arrows,positioning, shapes}

\newcount\Comments  
\definecolor{darkgreen}{rgb}{0,0.6,0}
\newcommand{\kibitz}[2]{\ifnum\Comments=1{\color{#1}{#2}}\fi}

\newcommand{\argy}[1]{\kibitz{red}{[ARGY:#1]}}

\providecommand{\keywords}[1]{\textbf{Keywords: } #1}

\begin{document}

\title{\vspace{-0.5cm}Optimizing Reachability Sets in Temporal Graphs by Delaying
\thanks{A preliminary conference version of this work appeared
in Proceedings of the 34th AAAI Conference on Artificial Intelligence (AAAI), New York, New York, USA, 2020.}}

\author{Argyrios Deligkas\thanks{%
Department of Computer Science, Royal Holloway University of London, UK. Email: \texttt{argyrios.deligkas@rhul.ac.uk}.
This work was done while the author was a postdoc at Materials Innovation Factory at the University of Liverpool, funded 
by LRC.} 
\and Igor Potapov\thanks{%
Department of Computer Science, University of Liverpool, UK. Email: \texttt{potapov@liverpool.ac.uk}. Research partially funded by the grant EP/R018472/1 ``Application driven Topological Data Analysis''.}}
\date{\vspace{-1.0cm}}

\maketitle


\begin{abstract}
A temporal graph is a dynamic graph where every edge is assigned a set of integer 
time labels that indicate at which discrete time step the edge is available. 
In this paper, we study how changes of the time labels, corresponding to delays 
on the availability of the edges, affect the reachability sets from given sources.
The questions about reachability sets are motivated by numerous applications of 
temporal graphs in network epidemiology, which aim to minimise the spread of 
infection, and scheduling problems in supply networks in manufacturing with
the opposite objectives of maximising coverage and productivity.
%
We introduce control mechanisms for reachability sets that are based on two natural
operations of delaying. 
%
The first operation, termed merging, is global and batches together consecutive time
labels into a single time label in the whole network simultaneously. This corresponds to postponing
all events until a particular time.
The second, imposes independent delays on the time labels of every edge of the 
graph.
We provide a thorough investigation of the computational complexity of different 
objectives related to reachability sets when these operations are used.
For the merging operation, i.e. global lockdown effect, we prove NP-hardness results for several minimization 
and maximization reachability objectives, even for very simple graph structures.
%
For the second operation, independent delays, we prove that the minimization problems 
are NP-hard when the number of allowed delays is bounded. 
We complement this with a polynomial-time algorithm for
minimising the reachability set in case of unbounded delays.
%

%

\end{abstract}


\section{Introduction}
\label{sec:intro}
A plethora of real life scenarios can be modelled as a dynamic network that changes over time. 
These scenarios range from train, bus, and distribution schedules, to livestock movements
between farms and virus spreading. Many of these dynamic networks consist of a fixed set of 
nodes and what changes over time is the connectivity between pairs of them; the locations of the
stations, distribution centers, and farms remain the same over time, while the connections 
between any two of them can change every few minutes, hours, or days. An equivalent way to see
these networks is as a sequence of static networks that change according to a predetermined, 
known in advance, schedule. 
These type of networks, known as {\em temporal graphs}, were formalized in the seminal work 
of~\cite{kempe02}. Since then, there is flourish of work 
on temporal graphs
\cite{zschoche_MFCS2018,Michail16,Casteigts2019,chen_et_al:LIPIcs:2018:9972,MICHAIL20161,Erlebach2015}.
Formally, in a temporal graph every edge is assigned a set of integer time labels that indicate 
at which discrete time steps the edge is available. In other words, a temporal graph 
is a schedule of edge sets $E_1, E_2, \ldots, E_t$ over a fixed set of vertices. 

In this paper, we study questions related to {\em reachability sets} on temporal graphs. 
The reachability set of vertex $v$ is the set of vertices that there exist {\em temporal paths}
from $v$ to them. 
Informally speaking, a (strict) temporal path 
is a path whose edges, 
additionally to the usual connectivity, use (strictly) increasing time labels~\cite{Whitbeck2012}. 
However, in contrast to static graphs, temporal paths do not preserve transitivity. As a result,
some well known concepts from graph theory, like Menger's theorem, do not hold for 
temporal graphs and have to be redefined~\cite{AKRIDA2019,Mertzios2013,erlebach_MFCS2018}.

Reachability sets emerge naturally in a wide range of models and real life applications \cite{IEEE2014,Potapov-2004,Niskanen2020}.
One of them, related to temporal graphs, is the minimization of spread of infectious diseases over contact networks. 
The recent case of coronavirus COVID-19 once again reveals the importance of  
effective prevention and control mechanisms \cite{COVID}.
Data provide significant evidence that commuter patterns and airline flights in case of infectious diseases in humans
\cite{colizza2006role,brockmann2013hidden},  and livestock movements between farms for the animal case 
\cite{mitchell2005characteristics} 
could spread infectious diseases. 
In particular, the impact of animal movements has been extensively studied within the 
epidemics community. Recent studies have shown that restricting of animal movements is 
among the most effective methods for controlling the spread of an 
infection~\cite{jones2019bluetongue,turner2019effect,turner2012modelling,thulke2011role,buhnerkempe2014impact} 
and changes in the network of animal movements between farms (nodes in a graph) can
significantly decrease the spread of an infection~\cite{gatescontrolling,mohr2018manipulation}.
Contrary to the applications above, there are cases where we wish to maximize the reachability sets.
Consider for example a distribution schedule, where the use of the underlying network for every day 
comes at a fixed cost. Then, the goal is to optimally utilize the infrastructure of the network by 
choosing which days to use the network in order to achieve maximum reachability while maintaining low
cost.

The importance of these problems combined with their inherit temporal nature, made 
temporal graph theory a tool for analyzing  epidemiology in (animal) networks
\cite{braunstein2016,COA2015,abs-1802-05905,Deleting_Edges,ENRIGHT201888,noremark2014,valdano2015analytical,valdano2015predicting}. 
In \cite{abs-1802-05905} and \cite{Deleting_Edges} the authors studied how reachability sets on temporal 
graphs change, when the schedule of the edges can change according to specific operations. More 
specifically, they studied \minmaxreach and \maxminreach problems where the goal, respectively, was to minimize 
the maximum, or to maximize the minimum size of reachability sets of a given set of sources.
In \cite{abs-1802-05905} these objectives were studied  under the operation of {\em reordering} of the edge 
sets and it was proven that both problems are \NP-hard. 
For \minmaxreach, \cite{Deleting_Edges} studied the operations of deletion of whole edge sets, or 
individual edges within an edge set. It was proven that \minmaxreach under both notions of deletion 
is \NP-hard, but they admit an FPT algorithm when they are parameterized by the treewidth of the network 
and the maximum permitted number of reachable vertices.

Although optimization of reachability sets capture important real life problems, some of the proposed 
solutions are not completely satisfying due to big changes in infrastructure.
For example, reordering of edge sets can be difficult, costly, or even impossible to perform due to the 
physical constraints of the network, or due to the number of changes required in the existing infrastructure. 
For edge-deletions, and delay-based operations studied in this paper, an upper bound on the number of allowed deletions is required. This is crucial, since 
the deletion of every edge trivially minimizes the reachability sets, but makes the existing network
infrastructure useless. In addition, edge deletions can create a ``bottleneck'' problem in the network even 
if their number is bounded. The deletion of an edge can create a sink to the network, or completely isolate 
some of its parts. 
Instead, we 
wish to study the following problem.
\begin{quote}
{\em Given a temporal graph and a set of sources, optimize the size of the 
reachability set of the sources using only {\em ``natural''} and {\em ``infrastructure-friendly''} 
operations.}
\end{quote}
\noindent
Natural operations should be intuitive and should be easy to apply; deletion, or postponement, of a 
temporal edge can naturally happen. On the other hand, bringing forward a temporal edge cannot be 
always feasible due to physical constraints and infrastructure constraints. Infrastructure-friendly operations should 
not be to difficult to perform and   should not require many changes to the given network and temporal 
schedule.

\paragraph{\bf Our contribution.}
Our contribution is twofold. Firstly, we introduce and study two operations, {\em merging} and 
{\em delaying}, that are natural and infrastructure-friendly and were not studied in the past.  
The idea behind both operations is the postponement of the edges of the graph. 
Merging operation is parameterized by $\lambda$ and it  batches together $\lambda$ consecutive 
edge sets; a $\lambda$-merge on $E_1, \ldots E_\lambda$ changes the first $\lambda-1$ edge sets to 
the empty sets, and the $\lambda$-th edge set is the union of all $\lambda$ edge sets. The delaying
operation, independently delays a temporal edge; a $\delta$-delay on the label $i$ of edge $uv$ changes
it to $i+\delta$. 
In contrast to deletion of temporal edges that can directly isolate vertices, 
our operations isolate some vertices {\em only temporarily}. Our second contribution is a
thorough investigation of the computational complexity of reachability objectives under these operations.
In particular, we study the \minmaxreach, \minreach, and \minavgreach, where in the last two cases the goal 
is to minimize the number, or the average number respectively, of reachable vertices from a given set of 
sources.
With respect to maximization objectives, we study \maxreach, \maxminreach, and \maxavgreach. 
We proved that these problems are \NP-hard under the merging operation even for very restricted classes 
of graphs; see Table~\ref{tab:results}. Our results can be extended even in the graphs are directed acyclic graphs, or even uni disk graphs.
For the delaying operation, we studied the minimization problems. They remain \NP-hard when we {\em bound} 
the number of times we are allowed to use this operation. We complement these results with a polynomial 
time algorithm for the case of the unbounded number of delays that works for any $\delta$.

\begin{table}
\begin{center}

\begin{tabular}{| l | c | c | c | c |}
\hline
Problem & Graph & Sources & Labels/ & Edges/ \\
 & Class &  & Edge & Step \\
\hline
\hline
$\minreach$ & Path & $O(n)$ & 1 & 3 \\ 
\hline
$\minreach$ & Tree &  &  &  \\
$\minmaxreach$ & $\Delta = 3$ & 1 & 1 & 1 \\ 
$\minavgreach$ &   & & &\\ 
\hline
$\maxreach$ & Path & $O(n)$ & 1 & 4 \\ 
\hline
$\maxreach$ & Bipartite &  &  &  \\
$\maxminreach$ & $\Delta = 3$ & 1 & 1 & 4 \\ 
$\maxavgreach$ &   & & &\\
\hline
$\maxreach$ & Tree &  &  &  \\
$\maxminreach$ & $\Delta = 3$ & 1 & 1 & 8 \\ 
$\maxavgreach$ &   & & &\\
\hline
\end{tabular}
\caption{\NP-hardness results for the $\lambda$-merge operation for any $\lambda \geq 2$.
$\Delta$ denotes the maximum degree of the graph; Labels/Edge denotes the maximum number of labels an edge can have, i.e., how many times it is available; Edges/Step denotes the maximum number of edges available at any time step.}
\label{tab:results}

\end{center}
\end{table}


\section{Preliminaries}

\paragraph
{\bf Temporal Graphs.} 
A {\em temporal graph} \tgraph is a pair of a (directed) graph $G=(V,E)$ and a 
function \tcal that maps every edge of $G$ to a list of time steps at which the
edge is available. The maximum time step, $\tmax$, on the edges of $G$ defines 
the {\em lifetime} of the temporal graph.
Another interpretation of a temporal graph, which is useful in our case, is to see 
it as a schedule of subgraphs, or edge-sets, $E_1, E_2, \ldots, E_{\tmax}$ of $G$,
known in advance and defined by the function $\tcal$; at time step $t$ function 
$\tcal$ defines a set $E_t \subseteq E$ of edges available for this time step.
We will say that an edge has the label $i$, if it is available at time step $i$. 
The size of a temporal graph \tgraph is $|V| + \sum_{t \leq \tmax} |E_t|$.
A {\em temporal path} in $\tgraph$ from a vertex $v_1$ to a vertex $v_k$ is a sequence of edges 
$v_1v_2, v_2v_3, \ldots, v_{k-1}v_k$  such that each edge $v_iv_{i+1}$ is available at time step $t_i$ and $t_i \leq t_{i+1}$ for every $i \in [k-1]$. A temporal path is {\em strict} if $t_i < t_{i+1}$. In what follows, unless stated otherwise, we only consider strict temporal paths.

\paragraph {\bf Reachability Sets.} 
A vertex $u$ is {\em reachable} from vertex $v$ if there exists a temporal 
path from $v$ to $u$. We assume that a vertex is reachable from itself.
It is possible that $u$ is reachable from $v$, but $v$ is not reachable from 
$u$. The reachability set of $v$, denoted $\reach(v,\tgraph)$, is the set of 
vertices reachable from $v$ in \tgraph at time $\tmax$.  
Given a temporal graph with lifetime $\tmax$ and a time step $t < \tmax$,  
the set $\reach_{t}(v,\tgraph)$ contains all the vertices reachable from $v$ 
by time $t$. 
The set $\reach(v,\tgraph)$ can be computed in polynomial time with respect to 
the size of $\tgraph$; it suffices to check whether there exists 
a temporal path from $v$ to every vertex in $V$, which can be done efficiently in polynomial time \cite{WC+14}.

\paragraph{\bf Merging.}
A {\em merging} operation on $\tcal$ postpones some of the edge-sets in a particular way. 
Intuitively, a merging operation ``batches together'' a number of consecutive edge-sets, making them all available at a later time step; see Figure~\ref{fig:example} for an example.
\begin{definition}[$\lambda$-merge]
\label{def:kmerge}
For every positive integer $\lambda$, a $\lambda$-merge of $E_i, E_{i+1}, \ldots, 
E_{i+\lambda-1}$, replaces $E_j = \emptyset$ for every $i \leq j < i+\lambda-1$ and 
$E_{i+\lambda-1} = \bigcup_{i \leq j \leq i+\lambda-1}E_j$.
\end{definition}
Thus, every merge corresponds to the global delay of events from times 
$i, i+1, \ldots, i+\lambda -1$ to the time $i+\lambda-1$.
We say that two $\lambda$-merges are {\em independent}, if they merge 
$E_i, \ldots E_{i+\lambda-1}$ and $E_j, \ldots E_{j+\lambda-1}$ respectively, and
$i +\lambda - 1 < j$. \footnote{Independent merges (delays) are commutative in respect to their application. 
For example let us  consider a temporal graph with $4$ nodes 
$x_1 {\overset{1}{\longrightarrow}} x_2 {\overset{2}{\longrightarrow}} x_3 {\overset{3}{\longrightarrow}} x_4$ 
and two non-independent merges $E_1, E_2$ and $E_2, E_3$. Following the order of   $E_1, E_2$ and then  
$E_2, E_3$ we will have a graph 
$x_1 {\overset{3}{\longrightarrow}} x_2 {\overset{3}{\longrightarrow}} x_3 {\overset{3}{\longrightarrow}} x_4$ 
and the alternative order of merges  $E_2, E_3$ and then  $E_1, E_2$ will give
$x_1 {\overset{2}{\longrightarrow}} x_2 {\overset{3}{\longrightarrow}} x_3 {\overset{3}{\longrightarrow}} x_4$, 
so it creates some ambiguity for non-independent merges.}

When it is clear from the context, instead of writing that we merge $E_i$ with $E_{i+1}$
we will write that we merge $i$ with $i+1$.
\begin{definition}[$(\lambda,\mu)$-merging scheme]
\label{def:merging-scheme}
For positive integers $\lambda$ and $\mu$, a {\em $(\lambda,\mu)$-merging scheme} applies
 $\mu$ independent $\lambda'$-merges on $E_1, E_2, \ldots, E_{\tmax}$, where $\lambda' \leq \lambda$.
A merging scheme is {\em maximal} if there is no other feasible $\lambda$-merge available after applying all the other merges in the scheme. %
\end{definition}

\begin{figure}[h!]
\begin{center}
  \includegraphics[scale=0.3]{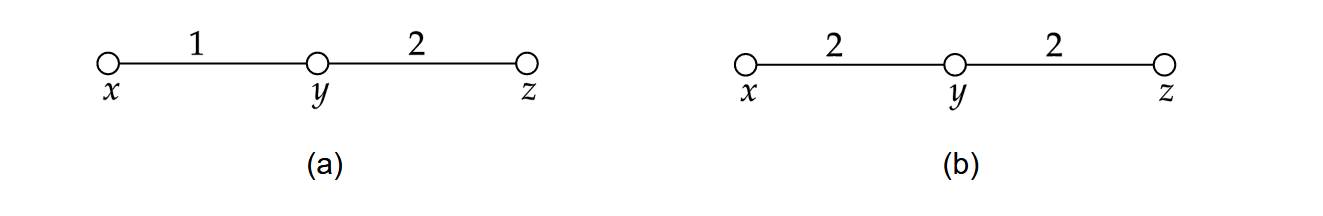}
  \caption{(a) A temporal graph where $E_1 = \{xy\}$ and $E_2 = \{yz\}$; 
  vertex $z$ is reachable from $x$. 
  (b) The resulted graph after merging $E_1$ with $E_2$; vertex $z$ is no
  longer reachable from $x$.}
    \label{fig:example}
  \end{center}
\end{figure}

A $(\lambda,\mu)$-merging scheme for a temporal graph \tgraph essentially 
produces a new temporal graph by modifying the schedule \tcal using $\mu$ 
independent $\lambda'$-merges, where $\lambda' \leq \lambda$. We will use $\mathcal{T}^M_{(\lambda,\mu)_i}$ to denote the modified schedule according to the merging scheme $(\lambda,\mu)_i$  and $\tuple{G, \mathcal{T}^M_{(\lambda,\mu)_i}}$ the corresponding modified temporal graph. When it is clear from the context, for notation brevity, we will use $\tcalm$ instead of $\mathcal{T}^M_{(\lambda,\mu)_i}$. A merging scheme is {\em optimal}  for an objective if it optimizes the reachability sets with respect to this objective.

Our goal is to compute merging schemes that optimise some 
objectives regarding reachability sets from a given set of vertices. 
The input to the problems we study consists of a temporal graph \tgraph, 
two positive integers $\lambda, \mu$, and a subset of vertices $S \subseteq V$ 
which we will term {\em sources}. The objectives  \minreach, \minmaxreach, \minavgreach, \maxreach, \maxminreach, \maxavgreach we study under $(\lambda,\mu)$-merging schemes are formally defined in the second column of Table~\ref{tab:prob}. 

Note that {\em any} merging operation monotonically decreases the size of the reachability set for {\em any} set of sources.
This does not impose any constraints when we are studying minimization objectives, since every extra merge can improve our objective. 
However, when we are studying maximization objectives under a $(\lambda, \mu)$-merging scheme we will require that the merging scheme will perform {\em at least} $\mu$ independent merges of length $\lambda$ each.

\paragraph{\bf Delaying.}
While merging operations affect globally the whole graph; {\em edge delays} affect only 
one label of an edge. We parameterize the delay operation by $\delta$; the maximum delay 
we can impose on a label of an edge. Hence, a $\delta$-delay on edge $uv$ at time step $i$
changes the label $i$ to $i' \leq i+\delta$. A $(\delta,\kappa)_i$-delaying scheme applies at most $\kappa$ $\delta$-delay operations on a schedule \tcal. We will denote $\mathcal{T}^D_{(\delta,\kappa)_i}$ the produced schedule. If we allow an unbounded number of delays, we will omit $\kappa$ and we will simply refer to such a scheme as $\delta$-delaying.
When it is clear from the context, for notation brevity, we will use $\tcaldk^D$ instead of $\mathcal{T}^D_{(\delta,\kappa)_i}$. The objectives we study are defined at the third column of Table~\ref{tab:prob}.

\probc{probc:minreach}
\probc{probc:minmax}
\probc{probc:minavg}
\probc{probc:maxreach}
\probc{probc:maxmin}
\probc{probc:maxavg}

\begin{table}
\begin{center}

\begin{tabular}{ l || l || l }
\hline
Problem & $\merging{\lambda,\mu}$ & $\delaying{(\delta,\kappa)}$ \\
\hline
\hline
 & \\[-0.75em]
 \ref{probc:minreach}. $\minreach$ &  $\min |\bigcup_{v \in S}\reach(v, \tuple{G,\tcalm})|$ &  $\min |\bigcup_{v \in S}\reach(v, \tuple{G,\tcaldk^D})|$\\ 
\hline
 & \\[-0.75em]
\ref{probc:minmax}. $\minmaxreach$ &  $\min \max_{v \in S} |\reach(v, \tuple{G,\tcalm})|$ &  $\min \max_{v \in S} |\reach(v, \tuple{G,\tcaldk^D})|$ \\ 
\hline
 & \\[-0.75em]
\ref{probc:minavg}. $\minavgreach$ &  $\min \sum_{v \in S} |\reach(v, \tuple{G,\tcalm})|$ &  $\min \sum_{v \in S} |\reach(v, \tuple{G,\tcaldk^D})|$ \\ 
\hline
 & \\[-0.75em]
\ref{probc:maxreach}. $\maxreach$ &  $\max |\bigcup_{v \in S}\reach(v, \tuple{G,\tcalm})|$ &  $\max |\bigcup_{v \in S}\reach(v, \tuple{G,\tcaldk^D})|$ \\ 
\hline
 & \\[-0.75em]
\ref{probc:maxmin}. $\maxminreach$ &  $\max \min_{v \in S} |\reach(v, \tuple{G,\tcalm})| $ &  $\max \min_{v \in S} |\reach(v, \tuple{G,\tcaldk^D})|$ \\ 
\hline
 & \\[-0.75em]
\ref{probc:maxavg}. $\maxavgreach$ &  $ \max \sum_{v \in S} |\reach(v, \tuple{G,\tcalm})|$ &  $\max \sum_{v \in S} |\reach(v, \tuple{G,\tcaldk^D})|$ \\
\hline
\end{tabular}
\caption{Problems~\ref{probc:minreach}~-~\ref{probc:minavg} are minimization 
problems, while Problems~\ref{probc:maxreach}~-~\ref{probc:maxavg} are maximization
problems. If $|S|=1$, then the solution for all minimization problems is the same; 
similarly for the maximization problems.}
\label{tab:prob}

\end{center}
\end{table}

\paragraph{\bf \maxsat.} To produce many of our results we use the problem \maxsat.
An instance of \maxsat is a CNF formula $\phi$ with $n$ Boolean variables and $m$ 
clauses, where each clause involves exactly two variables and every variable appears 
in at most three clauses.  The goal is to 
maximize the number of satisfied clauses. Without loss of generality we will 
assume that every variable in $\phi$ appears exactly one time as a positive literal 
and at most two times as a negative literal. In~\cite{BK99} it was proven that 
\maxsat is \NP-hard and that it is even hard to approximate better than 1.0005.

\section{Merging: Minimization problems}
\label{sec:min-hard}
In this section we study minimization problems under merging operations. 
To begin with, we prove that \minreach under $\merging{2,\mu}$ is \NP-hard even when $G$ 
is a path with many sources. Then, using this result, we explain how to get \NP-hardness for
any $\lambda \geq 2$.
Next, we show how to extend our construction and get a bipartite, planar graph of 
maximum degree 3 and only one source, and thus we prove \NP-hardness for \minreach, \minmaxreach, and \minavgreach. Our next result is \NP-hardness
for  the same set of problems on trees with one source. For this result, we derive
a new construction. 
Although all of our results are presented for undirected graphs, they can be trivially 
extended for directed graphs by simply adding both directions to the edges. However, we can extend them for DAGs (Directed Acyclic Graphs) too. For each one of our constructions we show how to get a DAG.

\subsection{\minreach on paths}
\label{sec:min-path}

\paragraph
{\bf Construction.}
We will reduce from \maxsat considering 
a CNF formula $\phi$
with $n$ Boolean variables and $m$ 
clauses.
The $k$-th clause of $\phi$ will be associated with the time steps $4k, 4k+1$ and
$4k+2$, while the $i$-th variable will be associated with the time steps 
$M+4i, M+4i+1$ and $M+4i+2$, where $M = 4m+4$.
For every clause of $\phi$ we construct a path with nine vertices, 
where the middle vertex is a source. 
Consider the path for the $k$-th clause, that involves the variables $x_i$ and $x_j$.
An example of such a path can be found on Figure~\ref{fig:min-tree-new}. 
The middle piece of the path consists of the vertices $c_k, y_k^l, z_k^l, y_k^r,
z_k^r$; where $c_k$ belongs to a set of sources, e.g. $c_k \in S$. Edge $c_ky_k^l$ has the label $4k$, edges $y_k^lz_k^l$ 
and $c_ky_k^r$ have the label $4k+1$, and edge $y_k^rz_k^r$ has the label $4k+2$. 
The labels on the left and the right pieces of the path depend on the literals of 
the variables of the clause. 
If variable $x_i$ appears in the clause with a positive literal, then we pick 
an arbitrary side of the path, say the left, and add the label $M+4i+1$ to the edge 
$z^l_kw^l_k$ and the label $M+4i+2$ to the edge $w^l_kv^l_k$. 
If $x_i$ appears in the clause with a negative literal, then we add the label 
$M+4i$ to the edge $z^l_kw^l_k$ and the label $M+4i+1$ to the edge $w^l_kv^l_k$.
In order to create a single path, we arbitrarily connect the endpoints of the paths
we created for the clauses; every edge that connects two such paths has label $1$. 
Observe that for the constructed temporal graph, \tgraph,
the following hold: 
\begin{enumerate}
    \item \tgraph has $9m$ vertices and lifetime $4m+4n+6$; 
    \item every vertex is reachable from one of sources; 
    \item at every time step there are at most three edges available; 
    \item every edge has only one label, i.e., it is available only at one time step.  
\end{enumerate}
Clearly, the size of \tgraph is polynomial to the size of $\phi$ and $|S|=m$.
We will ask for a $(2,m+n)$-merging scheme.

\begin{figure}
  \begin{center}
      \includegraphics[scale=0.30]{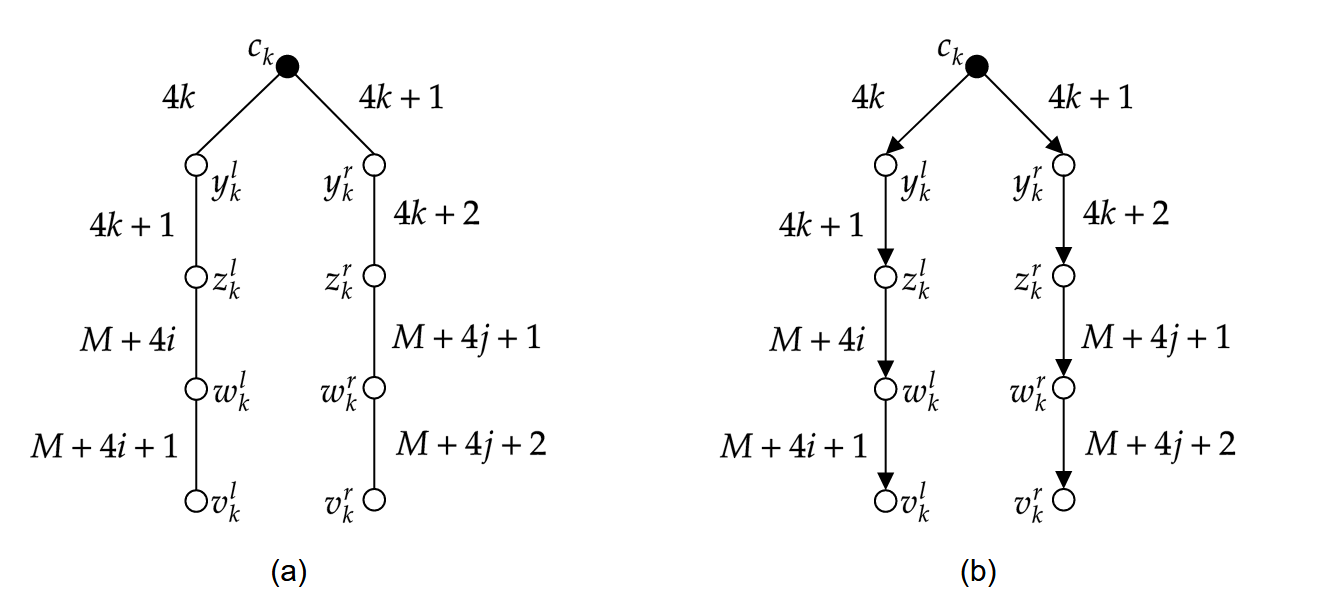}
  \end{center}
    \caption{(a) The gadget for the clause $(\bar{x}_i,x_j)$. The labels on the 
  edges denote the time steps these edges are available; $M = 4m+4$. 
  The solid vertex, $c_k$, is a source, i.e., it belongs to $S$. (b) The gadget for the directed case. The only difference from Subfigure (a) is the addition of the directions that can be used to prove \NP-hardness in DAGs.
  }
    \label{fig:min-tree-new}
\end{figure}

\paragraph{\bf Intuition.} The correctness of our reduction relies on two ideas.
The first idea is that in every subpath, under any $(2,\mu)$-merging scheme, {\em at most}
four vertices are not reachable from $S$. In order to make four vertices not 
reachable within a gadget, the following synergy must happen. Vertex $c_k$ should 
choose which side of the path to ``save''; merging $4k$ with $4k+1$ makes the vertices
$z^l_k, w^l_k, v^l_k$ unreachable; merging $4k+1$ with $4k+2$ makes the vertices
$z^r_k, w^r_k, v^r_k$ unreachable. Observe that only one of the two merges can
happen, since the two merges together are not independent. 
Hence, such a merge makes three vertices of one side unreachable. In order
to make the $v$-vertex of the other side unreachable, one extra merge has to 
happen. This merge will be translated to a truth assignment for a variable, which is
the second idea of our reduction. The merge of $M+4i$ with $M+4i+1$ corresponds to
setting $x_i$ to False and the merge of $M+4i+1$ with $M+4i+2$ to True. 
Note, that $M+4i+1$ can be merged {\em only} one of $M+4i$ and $M+4i+2$, since the merge of $M+4i$ with $M+4i+1$  and  the merge of $M+4i+1$ with $M+4i+2$ are {\em not} independent.

\begin{lemma}
\label{lem:min-path-sound}
If there exists an assignment that satisfies $l$ clauses of $\phi$,
then there exists a $(2,n+m)$-merging scheme such that $3m+l$ vertices are not 
reachable from $S$.
\end{lemma}
\begin{proof}
Given an assignment we produce a merging scheme as follows. If variable $x_i$ is
False, then merge $M+4i$ with $M+4i+1$; if it is True, we merge $M+4i+1$ with 
$M+4i+2$. Then, we consider every path that corresponds to a clause that is
satisfied and we check the side of this path that is no longer reachable from $c_k$, i.e. the side where the path was ``cut'' with respect to reachability from source $c_k$. 
For the $k$-th clause, if there is a cut from the side where the edge
with label $4k$ lies, then we merge $4k$ with $4k+1$; if the cut is from the side 
where the edge $4k+2$ lies, then we merge $4k+1$ with $4k+2$; if there are cuts in
both sides, or there are no cuts then we arbitrarily make one of the two merges. 
It is not hard to check that under this merging scheme: in every gadget that 
corresponds to a satisfied clause we have four vertices not reachable (we have $l$ such clauses, so $4l$ vertices are not reachable), and in
every gadget that corresponds to a non satisfied clause only three vertices are not
reachable (we have $m-l$ such clauses, so $3m - 3l$ vertices are reachable in this type of gadgets). Hence, given an assignment that satisfies $l$ clauses, we have made $m+n$ of 2-merges and there are $3m+l$ of unreachable vertices.
\end{proof}

\begin{lemma}
\label{lem:min-path-correct}
If there exists an optimal $(2,n+m)$-merging scheme such that $3m+l$ vertices are 
not reachable from $S$, then there exists an assignment for $\phi$ that satisfies 
$l$ clauses.
\end{lemma}
\begin{proof}
Firstly, observe that for the time steps $4k-1, 4k, 4k+1, 4k+2$, and $4k+3$, for 
every $k \in [m]$, no more than one 2-merge is necessary. This is because at time 
steps $4k-1$ and $4k+3$ there are no edges available, hence any merge that includes
these time steps does not change the set of reachable vertices. In addition, this
merge has to involve time step $4k+1$, else it is meaningless as we argued above.
The same holds for the time steps $M+4i-1, M+4i, M+4i+1, M+4i+2$, and $M+4i+3$ for every $i \in [n]$. Hence, in any optimal merging scheme there is only one 2-merge 
for every triple $4k, 4k+1, 4k+2$, where $k \in [m]$, and only one 2-merge for 
every triple $M+4i, M+4i+1, M+4i+2$, where $i \in [n]$. So, under any optimal merging
scheme, for every $k$, in the $k$-th gadget at least three vertices are not 
reachable (due to the merge that involves $4k+1$). In addition, at most four vertices 
are not reachable due to a, potential, extra merge that makes a $v$-vertex not 
reachable. 

In order to create a truth assignment for the variables of $\phi$, we proceed as
follows. Consider the $k$-th gadget and assume that involves the labels $M+4i$, $M+4i+1$, and $M+4i+2$.
If $M+4i+1$ is merged with $M+4i$, then we set $x_i$ to False. 
If $M+4i+1$ is merged with $M+4i+2$, then we set $x_i$ to True.
By the construction of the gadget, observe that $x_i$ satisfies the $k$-th clause. Any variables with undefined value, we set them to
True. So, if there are $l$ gadgets with four unreachable vertices, the constructed
assignment satisfies $l$ clauses of $\phi$.  To complete our proof, we need to argue
that the truth assignment for the variables is well defined, i.e. we did not set 
$x_i$ both to True and False. For contradiction assume that the value of $x_i$ is 
not well defined. This means that $M+4i+1$ is merged with $M+4i$, and that $M+4i+1$ 
is merged with $M+4i+2$ as well; a contradiction.
\end{proof}

The combination of Lemmas~\ref{lem:min-path-sound} and~\ref{lem:min-path-correct} already yield \NP-hardness for \minreach under $\merging{2,\mu}$, when $\mu$ is part of the input. 
\begin{corollary}
\minreach is \NP-hard under $\merging{2,\mu}$.
\end{corollary}

In the next theorem we explain how to get stronger \NP-hardness results for \minreach for a number of different parameters.
\begin{theorem}
\label{thm:main-min-hard}
\minreach under $\merging{\lambda,\mu}$ is \NP-hard for any $\lambda\geq 2$, even 
when $G$ is a path, and the following constraints hold:
\begin{enumerate}
\item every edge is available only at one time step;
\label{enum:one}
\item at any time step there are at most three edges available;
\item the merging scheme is maximal.
\label{enum:three}
\end{enumerate}
\end{theorem}
\begin{proof}
We start from the temporal graph $\tuple{G, \tcal}$ constructed above. Let $E_1. E_2, \ldots, E_{M+4n+2}$ be the edge-sets of $\tuple{G, \tcal}$. We will construct a new temporal graph $\tuple{G, \tcal'}$ as follows with edge-sets $E'_1, E'_2, \ldots$ where $E'_{k} = E_j$ if $k = \lambda \cdot j - 2$ for some $j \in \{1, 2, \ldots, M+4n+2 \}$, and $E'_{k} = \emptyset$ otherwise. Observe that in  $\tuple{G, \tcal'}$, any $\lambda'$-merge with $\lambda' < \lambda$ does not affect the set reachable from $S$ since it can merge at most one time step that contains edges from $G$. Hence, we can replace the 2-merges with $\lambda$-merges in  Lemmas~\ref{lem:min-path-sound} and~\ref{lem:min-path-correct} and get the proofs exactly in the same way as before. Again, observe that again at most three edges per time step are available. 

In addition, observe that in the proofs Lemmas~\ref{lem:min-path-sound} and~\ref{lem:min-path-correct} we have not used anywhere the fact that at most $\mu$ of $2$-merges were allowed. In other words, the hardness does not come from constraining the number of allowed $2$-merges, but from the way the 2-merges can be placed. Hence, in Lemma~\ref{lem:min-path-sound} we can arbitrarily extend the produced merging scheme to a maximal merging scheme. In Lemma~\ref{lem:min-path-correct},  without loss of generality
we can assume that the merging scheme is maximal, since this does not affect the correctness of the lemma.
\end{proof}

\subsection{Graphs with a single source}
\label{sec:min-one}
Next, we explain how to get \NP-hardness for the case with a single source, e.g. where $|S|=1$. 
This result immediately implies \NP-hardness for \minmaxreach and \minavgreach
under $\merging{\lambda,\mu}$, since when $|S|=1$ these problems coincide with 
$\minreach$.
It is not hard to get this result, given the previous construction.
We can simply add a vertex $c_0$ in the constructed graph that is the only source and 
it is connected with every vertex that used to be a source, so $c_0$ is connected with all 
the $c$-vertices from the previous construction, with an edge available at time step 1. 
Since the next time step that an edge exists in is time step 4, this means
that under {\em any} $(2,\mu)$-merging scheme, at time step 2 every vertex of $S$ has been 
reached by $c_0$. Hence, The \NP-hardness follows from the previous construction. However, 
this can be strengthen more by the design of a graph where every vertex has a  constant degree. 
Instead, we can modify \tgraph as shown in Figure~\ref{fig:min-one} and get a tree of maximum degree 3.

\begin{theorem}
\label{thm:one-min-hard}
\minreach, \minmaxreach, and \minavgreach under $\merging{\lambda, \mu}$ are
\NP-hard for any $\lambda \geq 2$ even when:
\begin{enumerate}
\item there exists only one source;
\item $G$ is a tree of maximum degree three and constant pathwidth;
\item every edge is available only at one time step;
\item at any time step there are at most three edges available;
\item the merging scheme is maximal.
\end{enumerate}
\end{theorem}

\begin{figure}
\begin{center}
  \includegraphics[scale=0.35]{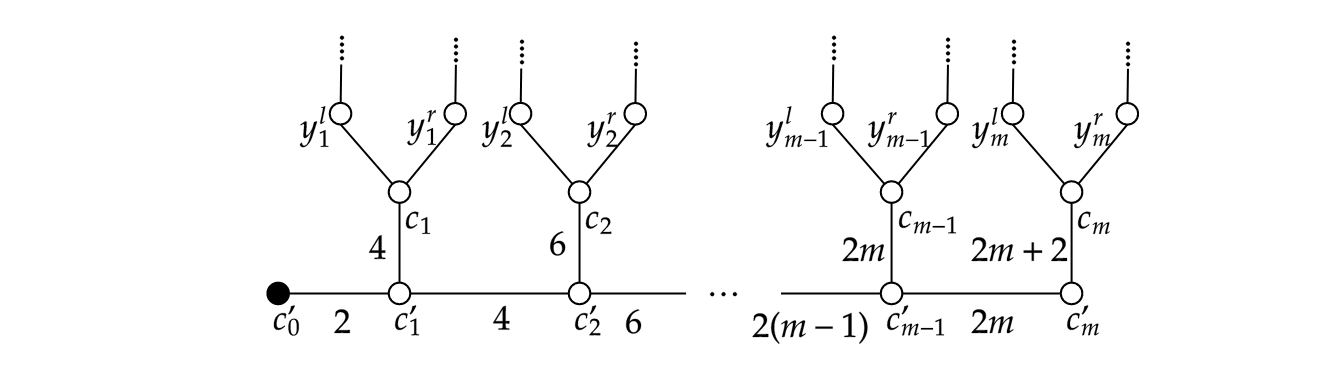}
  \caption{The construction for graphs with one source with $\lambda=2$. }
      \label{fig:min-one}
  \end{center}
\end{figure}

\begin{proof}
We will first prove the theorem for \minreach. Then the claims for \minmaxreach and \minavgreach  will immediately follow as all three objectives coincide when we have only one source.

We start from the instance used in the proof of Theorem~\ref{thm:main-min-hard}. 
Let $\tuple{G,\tcal}$ be the temporal graph constructed for the proof of Theorem~\ref{thm:main-min-hard} and let $E_1, E_2, \ldots, E_k$ denote its edge sets.
Recall, the instance there consists of $m$ paths where path $k$ has the source $c_k$. We will create a temporal graph $\tuple{G',\tcal'}$ where $G$ is a tree of maximum degree 3. Firstly, we create a path that consists of the vertices $c'_0, c'_1, \ldots, c'_m$ where $c'_0$ and $c'_m$ are the ends of the path. In addition, we will add the edges $c'_ic_i$ for every $i \in \{1, 2, \ldots, m\}$. Next we define the edge sets $E_1, E_2, \ldots, E_{\lambda \cdot (m+1) +k}$. 
\begin{itemize}
    \item If $t = \lambda$, then $E'_t = \{c'_0c'_1\}$.
    \item If $t = \lambda \cdot i$ where $i \in \{2, \ldots, m\}$, then $E'_t = \{c'_{i-1}c'_i, c'_{i-1}c_{i-1}\}$. 
    \item If $t = \lambda \cdot (m+1)$, then $E'_t = \{c'_mc_m\}$.
    \item If $t = \lambda \cdot (m+1) + j $ where $j \in \{1, \ldots, k\}$, then $E'_t = E_j$.
    \item In every other case $E_t = \emptyset$.
\end{itemize}
The vertex $c'_0$ will be the unique source of the constructed graph.
The crucial observation is that under {\em any} merging scheme that uses $\lambda$-merges, at time step $2m+\lambda$ every vertex $c_i$ for $i \in [m]$ will be reached. Hence, at time step $2m+\lambda$, under any $(\lambda,\mu)$-merging scheme, we get an instance that it is equivalent, with respect to reachability from this time step and on, to the instance of the previous
section. Then, we can use extactly the same arguments to derive the \NP-hardness. 
\end{proof}

\section{Merging: Maximization problems}
\label{sec:max-hard}
In this section we prove \NP-hardness for maximization problems.
Before we proceed with the exposition of the results though, we should 
discuss some issues about maximization reachability problems and merging. 
Any merge weakly decreases the reachability set from the sources. Hence, 
while in the minimization problems, in principle, we would like to perform 
as many merges as possible, for maximization problems we would like to 
perform the minimum number of merges that are allowed. In addition, the 
reachability set weakly decreases with $\lambda$, i.e., the size of the merge. 
Hence, for maximization problems it is better to do the smallest merge possible, 
i.e., perform only 2-merges. For this reason, if we get \NP-hardness for
$(2,\mu)$-merging schemes, then we can immediately conclude that finding
the optimal $(\lambda,\mu)$-merging scheme is \NP-hard, for any $\lambda \geq 2$.
So, for maximization problems, we {\em require} that {\em at least} $\mu$ 
$\lambda$-merges need to happen. This is motivated by distribution networks;
the use of the network comes at a cost, thus we would like to use 
the network as few times as possible.

Again, we prove our results by reducing from \maxsat. This time though
we need to be more careful; in order to make our reduction valid, we 
should not allow for ``dummy'' merges, i.e., merges that do not change
the reachable vertices from the sources. 
As before, we first prove \NP-hardness for paths with multiple sources. 
Then, we modify our reduction to get \NP-hardness for 
graphs with one source and we provide a reduction for trees with just 
one source.

\subsection{\maxreach on paths}
\label{sec:max-path}
\paragraph{\bf Construction.}
We will reduce from \maxsat.
Every variable $x_i$ of $\phi$ will be associated with the time
steps $3i-2, 3i-1$, and $3i$. For every variable we create a path of length 5 with ends the vertices $y^l_i$ and $y^r_i$; this path is depicted at Figure~\ref{fig:max-reach-path}(b). In this path, the edges at the ends of the paths have labels $3i$ and the rest of the edges have label $3i+1$. The paths that are created from variables will be termed {\em v-paths}. Both ends of every $v$-path are sources. For every clause we create a path of length 4.  So, if the $k$-th clause involves the variables $x_i$ and $x_j$, we construct a path with ends the vertices $v_k^l$ and $v_k^r$, and middle the vertex $c_k$, which will be termed $c$-vertex. The variable $x_i$ will be related to the sub-path between $v_k^l$ and $c_k$ and the variable $x_j$ will be related to the sub-path between $v_k^r$ and $c_k$. If $x_i$ appears with a positive literal in the clause, then the labels on the two edges that connect $v_k^r$ and $c_k$ are $3i-2$ and $3i-1$ respectively; else the
labels are $3i-1$ and $3i$. 
The top side of Figure~\ref{fig:max-reach-path}(a) depicts the path for the clause $(\bar{x}_i, x_j)$. 
 These paths will be termed {\em c-paths}. Both ends of every c-path are sources. To create one path, we connect the constructed paths through their endpoints. Namely, we set $y^r_i = y^l_{i+1}$ for every $i \in [n-1]$, we set $y_n^r=v_1^l$, and we set $v_k^r = v^l_{k=1}$ for every $k \in [m-1]$.
 
 The temporal graph \tgraph we have constructed has $5n+4m+1$, lifetime $3n+1$, and at any time step at most four edges available: if $t = 3i-2$, then there are four edges; if $t = 3i-1$, then there are three edges; and if $t = 3i$, then there are four edges.

\begin{figure}[ht]
\begin{center}
  \includegraphics[scale=0.35]{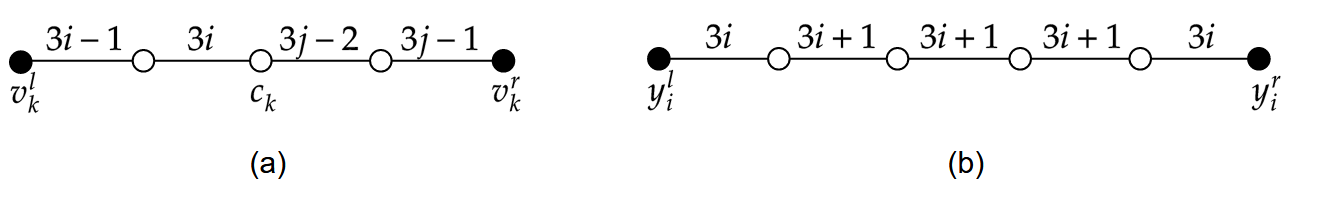}
  \caption{Gadgets for \maxreach on paths. Subfigure (a): c-path for the
  clause $(\bar{x}_i,x_j)$. Subfigure (b): v-path for variable $x_i$. }
  \label{fig:max-reach-path}
  \end{center}
\end{figure}
\paragraph{\bf Intuition.} The labels on the v-paths guarantee that there exists a $(2,n)$-merging scheme with the following properties: 
\begin{itemize}
    \item it maximizes the number of reachable vertices;
    \item it does not merge time step $3i$ with $3i+1$ for any $i \in [n]$.
\end{itemize} 
This guarantees two properties. Firstly, it guarantees that there exists an optimal $(2,n)$-merging scheme where all the vertices, except the $c$-vertices, are reachable from $S$. Hence, any optimal merging scheme, i.e., a merging scheme that maximizes the number of reachable vertices from $S$, {\em has} to maximize the number of reachable $c$-vertices. Second, given the subset of $c$-vertices that are reachable under the produced merging scheme, we can easily deduce a truth assignment that satisfies the clauses that correspond to the reachable $c$-vertices.

\begin{lemma}
\label{lem:max-path-sound}
If there exists an assignment that satisfies $l$ clauses of $\phi$, then there exists a $(2,n)$-merging scheme such that $5n+3m+1+l$ vertices are reachable from $S$.
\end{lemma}
\begin{proof}
Given an assignment for the variables of $\phi$ we produce a $(2,n)$-merging scheme as follows. If variable $x_i$ is false, then we merge $3i-2$ with $3i-1$; else we merge $3i-1$ with $3i$. Observe, all the vertices in the v-paths are  reachable under this merging scheme. Consider now the c-paths. Observe that all the vertices except the middle ones, the $c$-vertices, in these paths are reachable. Hence, we have to argue only about the $c$-vertices of \tgraph. Take a c-path that corresponds to a satisfiable clause; assume that this is the $k$-th clause that is satisfied by the variable $x_i$. Then, observe that the $c$-vertex of this path is reachable by the vertex whose side is related to $x_i$, say that it is $v_k^r$. If $x_i$ appears as a positive literal, then the path from $v_k^r$ to $c_k$ uses the labels $3i-2$ and $3i-1$, while the merging scheme merges $3i-1$ with $3i$. 
If $x_i$ appears as a negative literal, then the path from $v_k^r$ to $c_k$ uses the labels $3i-1$ and $3i$, while the merging scheme merges $3i-2$ with $3i-1$. Hence, we can conclude that $5n+3m+1+l$ vertices are reachable under the produced merging scheme.
\end{proof}

\begin{lemma}
\label{lem:max-path-correct}
If there exists an optimal $(2,n)$-merging scheme such that $6n+3m+1+l$
vertices are reachable from $S$, then there exists an assignment for $\phi$ that satisfies $l$ clauses.
\end{lemma}
\begin{proof}
To begin with, observe that it is without loss of generality to assume that the
$(2,n)$-merging scheme that maximizes the reachable set from $S$ does not merge
$3i$ with $3i$ for any $i \in [n]$. For contradiction, assume that there exits an $i$ such that every $(2,n)$-merging scheme that maximizes the reachable set from $S$ has to merge $3i$ with $3i+1$. Observe that this merge makes the two vertices in the middle of the $i$th v-path gadget unreachable. 
We will consider the following two cases for $3i-1$.
\begin{itemize}
    \item $3i-1$ is not merged. Then we can ``unmerge'' $3i$ and $3i+1$ and merge instead $3i-1$ with $3i$. Thus the number of 2-merges remains the same. In addition, the two vertices in the middle of the $i$th v-path are reachable. Finally, in the worst case, the vertices $c_k$ and $c_{k'}$ will become unreachable. These are the two vertices that correspond to the two clauses where variable $x_i$ appears with a negative literal. Hence, we get a different $(2,n)$-merging scheme with not decreased reachability set for $S$ that does not merge $3i$ with $3i+1$.
    \item $3i-1$ is merged. Then, there should exist a $j$ such that $3j-1$ that is not merged and at least one of $3j-2$ and $3j$ is not merged; this is due to the fact that the $(2,n)$-merging scheme that maximizes the set of reachable vertices uses exactly $n$ 2-merges. So, again we can unmerge $3i$ and $3i+1$  and instead add a merge that involves $3j-1$. This way, of reachable vertices is not decreased. By the unmerging the two vertices in the middle of the $i$-th vertex path are reachable. For the merge that involves $3j-1$ one of the following should hold.
    \begin{itemize}
        \item $3j-1$ is merged with $3j-2$. At most one vertex $c_k$ will become unreachable; this is the vertex that corresponds to the clause where variable $x_j$ appears as a positive literal. Thus, the initial $(2,n)$-merging scheme did not maximize the number of reachable vertices from $S$.
        \item $3j-1$ is merged with $3j$. Then, at most two vertices $c_k$ and $c_{k'}$ become unreachable. These vertices will correspond to the clauses where $x_j$ appears as a negative literal. Hence, we get a different $(2,n)$-merging scheme with the same number of reachable vertices that did not merge $3i$ with $3i+1$.
    \end{itemize}
\end{itemize}

From the above we can conclude that we have a $(2,n)$-merging scheme where $3i-1$ is merged with one of $3i-2$ or $3i$ for every $i \in [n]$.
So, given a $(2,n)$-merging scheme that satisfies the constraints above, we construct
the truth assignment for the variables of $\phi$ as follows. If $3i-1$ is merged
with $3i-2$, then we set $x_i$ to False; else we set it to True. Observe that for every c-path where vertex $c_k$ is reachable under the merging scheme, we get that the corresponding clause is satisfied by the produced assignment; this it due to our 
construction. It remains to show that there exist $l$ c-paths where the $c$-vertex is reachable. This is not hard to see. Firstly, all the vertices of the v-paths are reachable under the merging scheme. Then, for all clause-paths observe that under any $(2,n)$-merging
scheme all the vertices except the $c$-vertices are reachable; these are $3m+1$ in 
total. Hence, it must be true that under the assumed merging scheme $l$ $c$-vertices
are reachable from $S$.
\end{proof}

Lemmas~\ref{lem:max-path-sound} and~\ref{lem:max-path-correct} imply that \maxreach 
under $\merging{2,\mu}$ is \NP-hard. As we have already explained, this means that the 
problem is \NP-hard for any $\lambda \geq 2$.  

\begin{theorem}
\label{thm:max-path}
\maxreach under $\merging{\lambda,\mu}$ is \NP-hard for any $\lambda \geq 2$, even 
when the underlying graph is a path, every edge has one label, and at any time step 
there exist at most four edges available.
\end{theorem}

\subsection{Graphs with a single source}
\label{sec:max-one}
In this section we utilize the c-paths and the v-paths from the previous section and prove \NP-hardness for \maxreach with only one source. Thus, we get as a corollary \NP-hardness for \maxminreach and \maxavgreach. An easy way to prove \NP-hardness for this case, would be to shift all the labels of the edges of the v-paths and c-path by 1, create a vertex $v_0$ which would be the unique source, and connect it to every endpoint of the c-paths and v-paths with edges that appear at time step 1. However, this is not the strongest \NP-hardness result we could get. Instead, next we provide a very constrained class of graphs.

\paragraph{\bf Construction.}
We first create a path of length $2n+3m+7$ that consists of vertices $w_1, \ldots, w_{2n+3m+7}$, where $w_1$ is the unique source of the graph. For every $i \in [2,2n+3m+3]$, the label of edge $w_{i-1}w_i$ is $i$. The labels of the remaining edges are as follows: $w_{2n+3m+3}w_{2n+3m+4}$ has label $2n+3m+5$; $w_{2n+3m+4}w_{2n+3m+5}$ has label $2n+3m+8$; $w_{2n+3m+5}w_{2n+3m+6}$ has label $2n+3m+11$; $w_{2n+3m+6}w_{2n+3m+7}$ has label $2n+3m+14$.
We then use the c-paths and the v-paths from the previous section and connect them with the path we have created. More specifically, for the $i$-th v-path we connect  its left endpoint, $y_i^l$, with vertex $w_{2i-1}$ with an edge that has label $2i-1$ and its right endpoint, $y_i^r$, with vertex $w_{2i}$ with an edge that has label $2i$. For the $j$-th c-path, we connect endpoint $v^l_k$ with vertex $w_{2n+3j}$ with an edge that has label $2n+3j$ and endpoint $v^r_k$ with vertex $w_{2n+3j+2}$ with an edge that has label $2n+3j+2$. Finally, for every edge that belongs to a c-path or to a v-path and used to have label $i$  in the new construction it will have label $2n+3m+2+i$. Figure~\ref{fig:max-reach-graph} depicts how exactly this is done. Observe that the constructed temporal graph $\tuple{G, \tcal}$ has at most four edges available at any time step, $G$ is bipartite since every cycle in the graph has even length, and every edge of $G$ appears only at one time step. 

\begin{figure}
\begin{center}
  \includegraphics[scale=0.30]{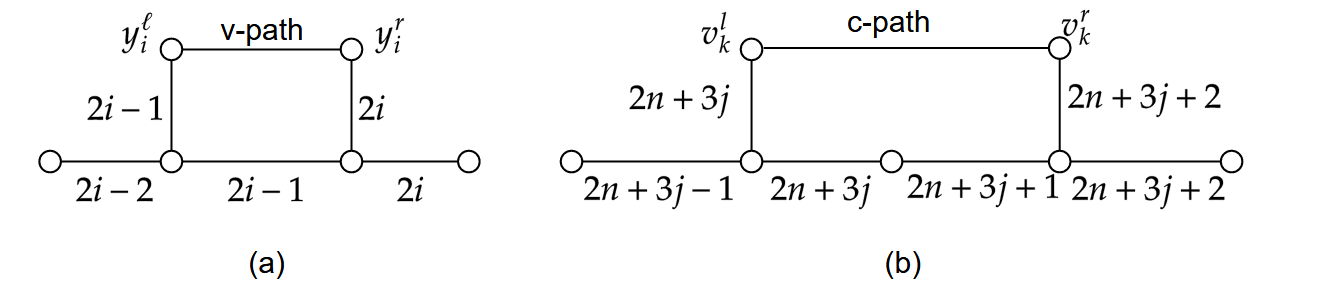}
  \caption{The gadgets used to prove \NP-hardness for \maxreach in graphs with only one source. Subfigure (a) depicts how we connect a v-path, while Subfigure (b) shows how we connect a c-path to the path we have created.}\label{fig:max-reach-graph}
  \end{center}
\end{figure}

\paragraph{\bf Intuition.} The high level intuition behind the following theorem is that {\em any} $(2,n)$-merging scheme that maximizes the vertices reachable from $w_1$ {\em does not} merge any edges with label less than $2n+3m+2$. This is because such a merge makes at least last five $w$-vertices unreachable while any other merge makes at most two vertices unreachable. Hence, at time step $2n+3m+2$ all the endpoints of c-paths and v-paths have been reached and thus we have an instance equivalent to the instance from Section~\ref{sec:max-path} which we know that it is \NP-hard to solve.

\begin{theorem}
\label{thm:max-one}
\maxreach, \maxminreach, \maxavgreach are \NP-hard under the merging operation even on temporal graphs $\tuple{G,\tcal}$ where 
\begin{itemize}
    \item there is one source;
    \item $G$ is bipartite or maximum degree 3;
    \item every edge appears only at one time step;
    \item at any time step there are at most four edges available.
\end{itemize}
\end{theorem}
\begin{proof}
Consider an arbitrary $(2,n)$-merging scheme that maximizes the reachability set in the graph $\tuple{G, \tcal}$ constructed above. Firstly, observe that this merging scheme does not merge any time step $i$ where $i \leq 2n+3m+2$. To see why this is the case, assume for the sake of contradiction that there exists a $(2,n)$-merging scheme $i$ and $i+1$ are merged, where $i \leq 2n+3m+2$. Then, observe that at every vertex $w_j$ with $j \geq i+1$ is not reachable. Hence, at least five vertices are not reachable. On the other hand, since the merging scheme has $n$ 2-merges, it means that there an $j \in [n]$ such that $2n+3m+2+3j$ is not merged with $2n+3m+2+3j+1$. If we merge these labels we make at most two vertices unreachable. This contradicts the fact that the original merging scheme was optimal. Hence, under any optimal $(2,n)$-merging scheme we get that after timestep $2n+3m+2$ every endpoint of a c-path and a v-path are reachable. In addition, observe that the vertices $w_i$ with $i > 2n+3m+2$ will be reachable independently from the merging scheme chosen as long as there are no merges before timestep $2n+3m+2$. Hence, after timestep $2n+3m+2$ we have an sub-instance that it is equivalent to the instance from Section~\ref{sec:max-path}; they are identical up to a shift of the labels by $2n+3m+2$. We know that \maxreach is \NP-hard for this instance, thus the theorem follows.
\end{proof}

\subsection{Trees with one source}
\label{sec:max-tree}
In this section we prove that \maxreach is \NP-hard even on trees with only one source.  Again, our reduction is from \maxsat. First, we  explain how to get  \NP-hardness for forests where every connected component of the forest has only one source. Then, using the idea from the previous section, we connect the components of the forest and create a single tree with only one source.

\paragraph{\bf Construction.} 
We will use a similar approach as before and we associate each variable with three consecutive time steps: variable $x_i$ will be associated to time steps $3m+3i-2$, $3m+3i-1$, and $3m+3i$, where $m$ is the number of the clauses. This time though we associate every clause with three consecutive time steps as well. So, the $k$-th clause with be associated with time steps $3k-2, 3k-1$ and $3k$.

\begin{figure}
  \begin{center}
  \includegraphics[scale=0.3]{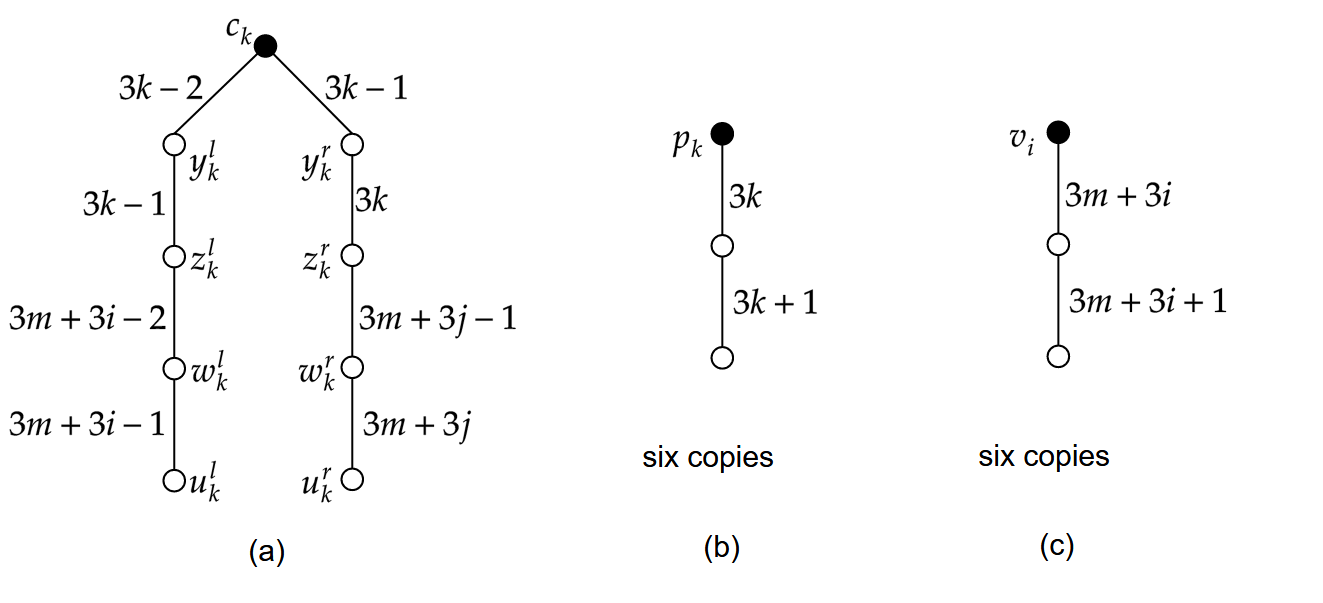}
  \caption{The gadgets used in Section~\ref{sec:max-tree} to show hardness for \maxreach on trees.}\label{fig:max-reach-tree-one}
  \end{center}
\end{figure}

For every clause we create a tree with 9 vertices where only one of them is a source;  an example of this tree is depicted in the left part of Figure~\ref{fig:max-reach-tree-one}. Each such tree consists of three pieces: the middle piece, the left piece, and the right piece. 
For the $k$-th clause, the middle piece consists of the vertices $c_k, y_k^l, z_k^l, y_k^r, z_k^r$; $c_k$ is a source. 
Edge $c_ky_k^l$ has the label $3k-2$, edges $y_k^lz_k^l$ and $c_ky_k^r$ have the label $3k-1$, and edge $y_k^rz_k^r$ has the label $3k$. The labels on the left and the right pieces of the path depend on the literals of the variables of the clause. 
If variable $x_i$ appears in the clause with a positive literal, then we pick an arbitrary side of the path, say the left, and add the label $3m+3i-2$ to the edge $z^l_kw^l_k$ and the label $3m+3i-1$ to the edge $w^l_kv^l_k$. 
If $x_i$ appears in the clause with a negative literal, then we add the label $3m+3i11$ to the edge $z^l_kw^l_k$ and the label $3m+3i$ to the edge $w^l_kv^l_k$.
In addition, we create the following. For every $k \in [m]$ we create six paths of 
length 2; see Figure~\ref{fig:max-reach-tree-one}(b). 
One endpoint of every such path is a source. The labels of the edges for every path, from the source to the other end, are $3k$ and $3k+1$. 
For every $i \in [n]$ we create six paths of length 2; the right part of Figure~\ref{fig:max-reach-tree-one}(c). One endpoint of every path is a source. The labels of the edges for every path, from the source to the other end, are $3m+3i$ and $3m+3i+1$. 
The constructed temporal graph \tgraph has $27m+18n$ vertices and lifetime $3m+3n+1$. All the vertices of the graph are reachable from $S$. We ask for a $(2, m+n)$-merging scheme. Observe that at any time step there exist at most eight edges available.

\paragraph{\bf Intuition.} 
The high level idea behind the construction is that in any $(2, m+n)$-merging scheme should make  all the vertices of the paths of length 2 are reachable; this is guaranteed by the 
number of copies of the paths and the choice of the labels in these paths. 
Hence, an optimal merging scheme maximizes the number of reachable vertices in 
the gadgets for the clauses. In every such gadget at most six vertices can reached. In order to reach 6 vertices the following synergy must happen: the root has to choose a side to reach and then the last part of that side has to indeed not be blocked by a merge. The fist step corresponds to a clause choosing which variable will satisfy it; our construction guarantees that can choose any of the two variables. The second step will happen only if the variable has the ``correct'' value for the clause.

The construction and the intuition behind this result is similar in spirit with the \NP-hardness result of Section~\ref{sec:min-path}. The main gadget for the clauses shares the same principle as in Section~\ref{sec:min-path}. There, vertex $c_k$ chooses via a merge which side to ``save'' and an extra merge at this side must happen in order to achieve this. Here, vertex $c_k$ chooses via a merge which side to reach and an extra merge to achieve this. 

\begin{lemma}
\label{lem:max-tree-sound}
If there exists an assignment for the variables of $\phi$ that satisfies $l$ clauses,
then there exists a $(2,m+n)$-merging scheme for \tgraph such that $26m+18n+l$
vertices are reachable from $S$.
\end{lemma}

\begin{proof}
Given an assignment that satisfies $l$ clauses we create the a merging scheme as 
follows. If $x_i$ is True, then we merge $3m+3i-1$ with $3m+3i$, else we merge 
$3m+3i-2$ with $3m+3i-1$. In addition, we consider every clause separately. If the $k$-th
clause is satisfied by the variable that corresponds to the left piece of the gadget
we created for the clause, then we merge $3k-2$ with $3k-1$; else we merge $3k-1$ with 
$3k$. If the clause is not satisfied, then we merge $3k$ with $3k-1$. Observe that
under this merging scheme all the vertices of the length-two paths are reachable.
In addition, in every tree that corresponds to a clause at least three vertices of
$y^l_k, y^r_k, z^l_k$, and $z^r_k$ are reachable from $c_k$. In addition, for every
satisfiable clause the vertices $w^l_k, v^l_k$, and $u^l_k$ are reachable if the 
clause is satisfied by the ``left'' variable, or the vertices $w^r_k, v^r_k$, and 
$u^r_k$ are reachable if the clause is satisfied by the ``right'' variable.
It is not hard to see that under this merging scheme $26m+18n+l$ vertices are 
reachable from $S$.
\end{proof}

To prove the other direction, we will first prove an auxiliary lemma.

\begin{lemma}
\label{lem:max-tree-aux}
Under any reachability-maximizing $(2,m+n)$-merging scheme  for the constructed temporal graph \tgraph, the following hold:
\begin{itemize}
     \item $3k-1$ is merged for every $k \in [m]$;
     \item $3m+3i-1$ is merged for every $i \in [n]$.
\end{itemize}
\end{lemma}

\begin{proof}
In order to prove the lemma, it suffices to prove that there does not exist a reachability-maximizing $(2,m+n)$-merging scheme that merges
\begin{itemize}
    \item $3k$ with $3k+1$ for any $k \in [m]$;
    \item $3m+3i$ with $3m+3i+1$ for any $i \in [n]$.
\end{itemize}
 Observe that any such merge makes six vertices of \tgraph unreachable; these vertices belong to a subset of paths of length 2 we constructed. For the sake of contradiction, assume that we have a reachability-maximizing $(2,m+n)$-merging scheme that merges $3k$ with $3k+1$; for the case where $3m+3i$ is merged with $3m+3i+1$ identical arguments apply. Then, at least one of the following cases is true:
 \begin{itemize}
     \item there exists a $k' \in [m]$, where $k' \neq k$ such that $3k'-1$ is not merged and can be merged with $3k'-2$ or $3k'$;
     \item there exists a $i \in [n]$, such that $3m+3i-1$ is not merged and can be merged with $3m+3i-2$ or $3m+3i$.
 \end{itemize}
 One of the two cases will be true due to the fact that any reachability-maximizing scheme will use exactly $m+n$ 2-merges and the lifetime of the graph is $3m+3n+1$. If the first case is true, then we can unmerge $3k$ with $3k+1$ and merge $3k'-1$ instead with $3k'$ say. Then, we make all six vertices from the paths of length 2 reachable, while we make at most five vertices unreachable at the gadget that involves timestep $3k'-1$: $c_{k'}$ will reach at least the vertices $y_{k'}^r, y_{k'}^l$ and $z_{k'}^l$. Thus we have increased the number of reachable vertices which contradicts the optimality of the initial merging scheme. If the second case is true then we unmerge $3k$ with $3k+1$ and merge $3m+3i-1$ instead with one of $3m+3i-2$ or $3m+3i$. Then, the only vertices that were reachable under the previous merging scheme, but are not longer reachable under the new merging scheme are the $u$-vertices of clause gadgets that are adjacent to an edge with label $3m+3i-2$ or label $3m+3i$. Observe that there are at most two edges with either of the cases. Hence, again this change increased the number of reachable vertices, a contradiction.
\end{proof}

Using Lemma~\ref{lem:max-tree-aux}, we are ready to prove the correctness of our construction.

\begin{lemma}
\label{lem:max-tree-correct}
If there exists a reachability-maximizing $(2,m+n)$-merging scheme such that 
$26m+18n+l$ vertices are reachable from $S$ in \tgraph, then there exists a
truth assignment that satisfies $l$ clauses of $\phi$.
\end{lemma}

\begin{proof}
From Lemma~\ref{lem:max-tree-aux} we know that $3k-1$ will be merged for every $k \in [m]$. Consider the following two cases for $3k-1$.
\begin{itemize}
    \item $3k-1$ is merged with $3k-2$. Then the vertices $z_k^l, w_k^l$, and $u_k^l$ become unreachable, while the vertices $y_k^l$, $y_k^r$, $z_k^r, w_k^r$, $u_k^r$, and $w_k^r$ are reachable {\em independently} the remaining merges.
    \item If $3k-1$ is merged with $3k$, then the vertices $z_k^r, w_k^r$, and $u_k^r$ become unreachable, while the vertices $y_k^r$, $y_k^l$, $z_k^l, w_k^l$, $u_k^l$, and $w_k^l$ are reachable {\em independently} the remaining merges.
\end{itemize}
Hence, in any reachability-maximizing $(2,m+n)$-merging scheme we get that at least five vertices from every clause-gadget are reachable. Furthermore, the merging scheme guarantees that all the vertices that belong to paths of length 2 will be reached. Thus, we have that $17m+9n$ will be reached. 

Hence, there exist $l$ $u$-vertices that are reachable the reachability-maximizing merging scheme we consider. If a $u$-vertex is reachable and is adjacent to an edge with label $3m+3i-1$, then we set variable $x_i$ to True; else if it is reachable and is adjacent to an edge with label $3m+3i$, then we set variable $x_i$ to False. If there are any variables which we have not defined their value, then we arbitrarily set the to True. Observe that this procedure sets correctly the values of the variables; no variable can get more than one value. This is because we cannot have two $u$-vertices that are adjacent to edges with labels $3m+3i-1$ and $3m+3i$ and both of them to be reachable. At least one of them will become unreachable since $3m+3i-1$ is merged. Our construction guarantees that this assignment for $x_i$ will satisfy all the clauses where the corresponding $u$-vertices are reachable. Since there are $l$ $u$-vertices reachable, then we get that $l$ clauses will be satisfied.
\end{proof}

The combination of the lemmas above already yield \NP-hardness for \maxreach on
forests. However, we can use a construction similar to one used in the proof of Theorem~\ref{thm:max-one} and get the following theorem.

\begin{theorem}
\label{thm:max-tree-one}
\maxreach, \maxminreach, \maxavgreach are \NP-hard under merging operations even when there exists only one source, $G$ is a tree of constant pathwidth, has maximum degree three, every edge has only one label, and at any time step there are at most eight edges available.
\end{theorem}

\begin{proof}
We begin by creating a path of length $7m+6n+10$: the path begins with $m$ $c'$-vertices, then we have $6m$ vertices that termed $p'$-vertices, then we have $6n$ $v'$-vertices, and lastly we have 10 $q'$-vertices. The labels are defined as follows: 
\begin{itemize}
    \item edge $c'_ic'_{i+1}$ has label $i$ for every $i \in [m-1]$;
    \item edge $c'_mp'_1$ has label $m$;
    \item edge $p'_ip'_{i+1}$ has label $m+i$ for every $i \in [3m-1]$;
    \item edge $p'_{3m}v'_1$ has label $7m$;
    \item edge $v'_iv'_{i+1}$ has label $7m+i$ for every $i \in [n-1]$;
    \item edge $v'_{3n}q'_1$ has label $7m+6n$;
    \item edge $q'_iq'_{i+1}$ has label $7m+6n+3i-1$ for every $i \in [9]$.
\end{itemize}
The unique source will be the vertex $c'_1$.
To complete the construction, we will use the gadgets we have constructed before. For every $k \in [m]$, we add the edge $c'_kc_k$ with label $k$, where $c_k$ is the vertex of the clause gadget we have created before. The $i$-th $p'$-vertex is connected with one $p$-vertex from the previous gadgets via an edge with label $m+i$, such that all $p$-vertices are connected to a $p'$-vertex. Furthermore, the $i$-th $v'$-vertex is connected to one $v$-vertex via an edge with label $7m+i$. Finally, we shift the labels of the gadgets from Figure~\ref{fig:max-reach-tree-one} by $7m+6n$, i.e. if an edge had the label $i$, now it will have the label $7m+6n+i$. Observe that the constructed graph is a tree of maximum degree 3, every edge has exactly one label, and at any time step there are at most eight edges available.

\begin{figure}
  \begin{center}
  \includegraphics[scale=0.35]{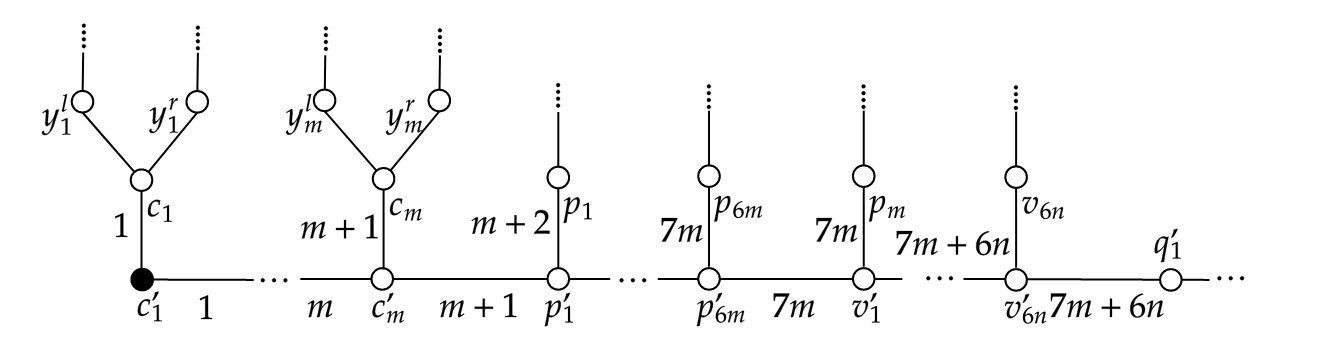}
  \caption{The tree with the one source for the \maxreach \NP-hardness.}\label{fig:max-reach-tree-one}
  \end{center}
\end{figure}

The crucial observation is that any reachability-maximizing $(2,m+n)$-merging scheme for the constructed graph will not merge any $i \leq 7m+6n$ since this will make at least the ten $q'$-vertices unreachable, while any other merge would make unreachable at most six vertices (if we merge $M+3k$ with $M+3k+1$ for $k \in [m]$, or if we merge $M+3m+3i$ with $M+3m+3i+1$ for some $i \in [n]$, where $M = 7m+6n$). We omit this part of the proof since it is almost identical to the proof of  Theorem~\ref{thm:max-one}. Hence, under any reachability-maximizing $(2,m+n)$-merging scheme, at time step $7m+6n$ every $c$-vertex, every $p$-vertex, and every $v$-vertex will be reached. In addition, all the $q'$-vertices will be reached under any such merging scheme, hence they do not affect the choices of the merging scheme. The correctness of the theorem follows by the observation that we have an sub-instance that it is identical to the \NP-hard instance we had before.
\end{proof}

\section{Delaying}
In this section we study {\em edge-independent} delays. Firstly, we show that when 
the number of allowed delays is bounded, then the minimization problems are \NP-hard.
Then, we study the case where the number of $\delta$-delaying operations is unbounded.
In contrast to unbounded edge deletions where the solution to the problems become trivial 
by essentially isolating every source, an unbounded number of independent $\delta$-delays, i.e when an edge cannot be delayed more than $\delta$ time steps, does not
trivialize the problems and most importantly does not destroy the underlying network.

\begin{theorem}
\label{thm:delay}
\minreach, \minmaxreach, and \minavgreach are \NP-hard under $\delaying{\delta}$, 
for any $\delta \geq 1$, when the number of operations is bounded by $\kappa$. 
In addition, they are $W[1]$-hard, when parameterized by $\kappa$.
\end{theorem}
\begin{proof}
To begin with, we make the following observation that will be helpful later in 
the proof.
For any temporal graph delaying any edge with label $\tmax$ does
not decrease the reachability set of any source. This is because if vertex $u$ 
from edge $uv$ is already reached by a source before $\tmax$, then $v$ will be 
reached from $uv$ independently from any delay on this edge. Hence, if we have 
a temporal graph with $\tmax = 2$, we can assume without loss of generality that 
in any optimal solution we delay edges with label 1 {\em only}. 

Firstly, we study $\minreach~ \delaying{\delta, \kappa}$. We will prove that 
the problem is hard even when $|S|=1$, i.e., there is only one source.
As before, this result implies that \minmaxreach and \minavgreach are \NP-hard 
under delaying as well.

To prove our result we follow a similar technique as the authors of~\cite{Deleting_Edges} 
and we reduce from \clique. An instance of \clique consists of a graph $G'=(V',E')$ and 
an integer $k$. We want to decide if $G'$ contains a clique of size $k$, i.e., a complete 
subgraph with $k$ vertices.
We construct a temporal graph \tgraph, where $G=(V,E)$, as follows. For every vertex and every
edge of $G'$ we create a vertex in $G$; we term the former ones $v$-vertices and the latter 
$e$-vertices. In addition, we create the vertex $v^*$ which will be the the only source.
Hence, $|V|=|V'|+|E'|+1$. 
The edge-set $E$ consists of:
\begin{itemize}
    \item $v^*v$, where $v$ is a $v$-vertex;
    \item $vu$, where $v$ is a $v$-vertex, $u$ is a $e$-vertex and the corresponding edge from 
    $E'$ contains the vertex $v$.
\end{itemize}
Every edge in $G$ that contains $v^*$ has label 1 and the rest of the labels have label 2. 
Finally, we set $\kappa = k$ and $\delta \geq 1$. We claim that $G'$ contains a clique of 
size $k$ if and only if there is a $(\delta,\kappa)$-delaying scheme such that $|V|-\frac{k(k-1)}{2}$ vertices are reachable from $v^*$ in $G$.

Assume now that $G'$ contains a clique $k$ and let $X = \{v_1, \ldots, v_k\}$ denote the set
of its vertices and let $Y$ denote the edges of the clique. Then, we delay by $\delta$ all 
the edges $v^*v$, where $v \in X$. We claim that no vertex in $Y$ is reachable from $v^*$.
To see why this is the case, let $y \in Y$ that correspond to the edge $v_iv_j$ where 
$v_i \in X$ and $v_j \in Y$. Observe that by the construction of \tgraph there are only two 
temporal paths from $v^*$ to $y$: $v^* {\overset{1}{\longrightarrow}} v_i {\overset{2}{\longrightarrow}} y$
and $v^* {\overset{1}{\longrightarrow}} v_j {\overset{2}{\longrightarrow}} y$. Hence, after
delaying the edges $v^*v_i$ and $v^*v_j$, we see that none of these paths is valid any more.
Since $|Y|=\frac{k(k-1)}{2}$, we get that $v^*$ reaches $|V|-\frac{k(k-1)}{2}$ vertices.

To prove the other direction, recall that since the \tgraph has lifetime 2, it suffices to delay
only edges with label 1. So assume that we can delay $\kappa=k$ edges with label 1 by $\delta$ such 
that $|\reach(v^*,\tuple{G,\tcaldk^{D}})| = |V|-\frac{k(k-1)}{2}$. Let $v^*v_1, \ldots,v^*v_k$ be the edges we
delayed. Observe that all the $v$-vertices are reachable by $v^*$. Hence, after the delays, $\frac{k(k-1)}{2}$ 
$e$-vertices are not reachable by $v^*$. So, by the construction of \tgraph, in $G'$ there should 
be $\frac{k(k-1)}{2}$ edges between the vertices $v_1, \ldots, v_k$. Hence, $G'$ contains a clique of size $k$.
This establishes the \NP-hardness of $\minreach~ \delaying{\delta, \kappa}$. Since $|S|=1$, the
hardness for the other two problems follows. Finally, since \clique is $W[1]$-complete and we have 
a parameterized $m$-reduction, the get the $W[1]$-hardness for the problems when they are parameterized
by the number of allowed delays $\kappa$.
\end{proof}

\subsection{A polynomial-time algorithm}
Next, we provide a polynomial-time algorithm for the minimization problems under $\delaying{\delta}$, for 
any $\delta \geq 1$.
Let us define the {\em reachability network} which will be used by our algorithm. Given a temporal
graph \tgraph with lifetime \tmax and a set of sources $S$, for every $t \leq \tmax$, we define 
$RV_t(\tgraph, S)$ to be the set of vertices that are reached at time $t$ {\em for first time} from
a vertex of $S$. Put formally, $v \in RV_t(\tgraph, S)$, if:
\begin{itemize}
    \item there exists an $s^* \in S$ such that the {\em earliest arrival} path from $s$ to $v$ arrives at time $t$;
    \item for every $t'<t$ there is no path from any $s \in S$ to $v$ that arrives at time $t'$.
\end{itemize}
Additionally, we define $RE_t(\tgraph, S)$ to be the set of temporal edges with label $t$ that are adjacent to vertices in  $RV_t(\tgraph, S)$. An example of these notions is depicted in Figure~\ref{fig:alg-example}.
Observe that we can decide if $v \in RV_t(\tgraph, S)$  via computing 
the earliest arrival paths between every $s \in S$ and $v$, which can be done in polynomial time with 
respect to the size of \tgraph~\cite{WC+14}. Similarly, we can efficiently compute $RE_t(\tgraph, S)$.
Finally, assume that the edge $uv$ {\em before any delaying} had the label $x$ and after some delays it has the label $t \leq x + \delta$. We say that it is $\delta$-possible to change a label of $uv$ from $t$ to $t+1$ if the $t+1 - x \leq \delta$; i.e. the edge is not delayed more than $\delta$ time steps.

\begin{figure}
			\begin{algorithm}[H]
				\SetKwInOut{Input}{Input}
				\SetKwInOut{Output}{Output}
				\DontPrintSemicolon
				
                $\tcald^D \leftarrow \tcal$;\;
				\For{$1\leq t < \tmax+ \delta$}{
				Compute $RE_{t}(\langle G,\tcald^D \rangle, S)$\;
					\ForEach{edge $e \in RE_{t}(\langle G,\tcald^D \rangle, S)$}{
					Change the label $i$ of $e$ to $t+1$, if this is $\delta$-possible;\;
					Update $\tcald^D$;\;
					}
				}
				
				\caption{Algorithm for $\delaying{\delta}$.}
				\label{alg:minreach-delay}
			\end{algorithm}
\end{figure}
	
\begin{figure}
  \begin{center}
  \includegraphics[scale=0.35]{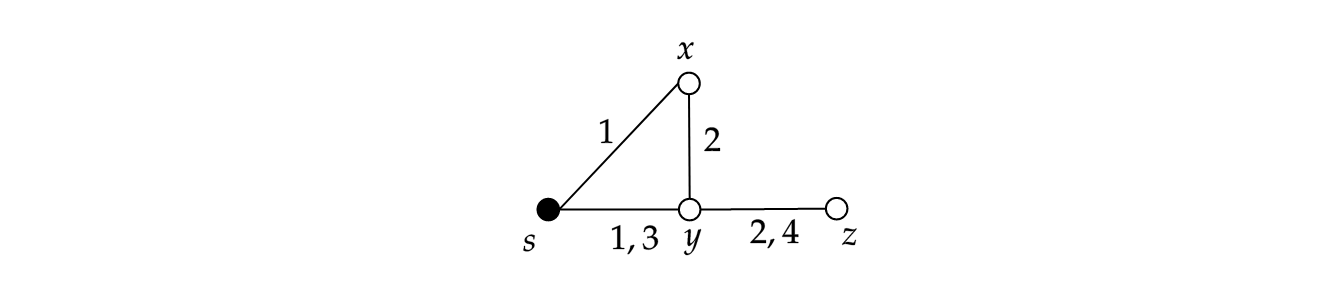}
  \caption{In the temporal graph \tgraph above with $S=s$, we have $RV_1(\tgraph,s) = \{x,y\}$,
  $RV_2(\tgraph,s) = \{z\}$, $RE_1(\tgraph,s) = \{sx,sy\}$, and $RE_2(\tgraph,s) = \{yz\}$. Any other set is empty.}
\label{fig:alg-example}
  \end{center}
\end{figure}

\begin{lemma}
\label{lem:alg-correct}
Algorithm~\ref{alg:minreach-delay} is optimal for \minreach, \minmaxreach, \minavgreach under 
$\delaying{\delta}$.
\end{lemma}
\begin{proof}
To prove the lemma, we will prove by induction that at any time step $1 \leq t < \tmax$ Algorithm~\ref{alg:minreach-delay} minimizes $\sum_{i=1}^t RV_{i}(G,\tcald^D, S)$, i.e., it minimizes the number of reachable vertices until time $t$. 
For $t=1$ the claim clearly holds. Every vertex reachable at time step 1 will be reached unless all of the edges with label 1 adjacent to a source get delayed. Algorithm~\ref{alg:minreach-delay} indeed delays all these edges, hence the claim holds.
For the induction hypothesis, assume that the lemma holds up to time step $t-1$. This means that $\sum_{i=1}^{t-1} RV_{i}(G,\tcald^D, S) \leq \sum_{i=1}^{t-1} RV_{i}(G,\hat{\tcald}, S)$ where $\hat{\tcald}$ is any other $\delta$-delaying scheme. 
In order to prove the lemma it suffices to prove that $RV_{t}(G,\tcald^D, S)\leq RV_{t}(G,\hat{\tcald}, S)$, since the correctness will follow from the induction hypothesis. 
To see why this is the case we observe that the algorithm will delay all the edges it can with label $t$ until time step $t+1$. Hence, the set of reachable vertices will increase only by the minimum possible number; this number is dictated by the number of edges that it is not $\delta$-possible to delay. 
\end{proof}

With Lemma~\ref{lem:alg-correct} in hand we can prove the following theorem.
\begin{theorem}
\label{thm:algo-delay}
\minreach, \minmaxreach, \minavgreach can be solved in polynomial time under $\delaying{\delta}$.
\end{theorem}
\begin{proof}
Lemma~\ref{lem:alg-correct} shows that Algorithm~\ref{alg:minreach-delay} is indeed optimal for \minreach, \minmaxreach, \minavgreach under $\delaying{\delta}$. It remains to show that the algorithm indeed requires polynomial time with respect to the input temporal graph. This is easy to see. At every iteration the algorithm checks the labels of the edges and increases the label of a subset of them. This clearly can be done in linear time in the size of the graph. In addition, we have $t_{\max}+\delta$ iterations, hence the overall running time of the algorithm is polynomial in the size of the input temporal graph.
\end{proof}

\section{Discussion}
In this section we discuss some further results results that can be easily derived from our \NP-hardness reductions. These include some inapproximability results and \NP-hardness for directed acyclic graphs and unit-disk graphs. In addition to these we highlight several challenging problems our work creates or leaves open.

\paragraph{\bf Inapproximability results.}
Our hardness results immediately imply, or can be easily extended to prove, several other interesting results. Firstly, we observe that all of our reductions under the merging operations are approximation preserving, thus since we use \maxsat, we get that there are no approximation schemes for these problems, unless $\PP = \NP$. More formally, we get that there exists an absolute constant $c$ such that the objectives we study are \NP-hard even to approximate better than $c$ under merging operations. A natural question is to ask whether we can find a polynomial-time algorithm with constant approximation
\begin{quote}
{\em Is there a polynomial-time algorithm for reachability objectives that achieves constant approximation under the merging operation? If not, can we get any non-trivial approximation in polynomial time?}
\end{quote}
We believe this is a challenging question due to the temporal nature of the model and thus novel techniques are required in order to tackle it. For example, there are instances where the optimal merging scheme uses $k$ merges and makes almost the whole graph not reachable, but  any subset of these merges does {\em not} decrease the reachability set at all, while any other merge makes only a constant number of vertices unreachable; see Figure~\ref{fig:greedy-bad}. This indicates that a good approximation algorithm cannot consider the merges independently.

\begin{figure}
\begin{center}
  \includegraphics[scale=0.30]{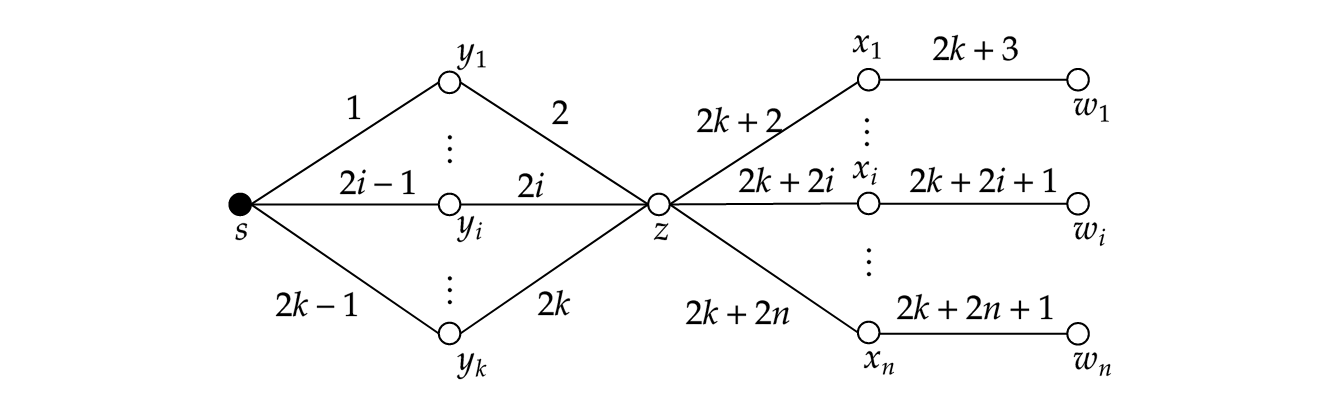}
  \caption{An example where the $(2,k)$-merging scheme that minimizes the reachability set of $s$ merges label $2i-1$ with label $2i$  for every $i \in [k]$. This results to a reachability set of size $k+1$. Any strict subset of these merges does not make any vertex unreachable, while any other 2-merge can make at most one vertex vertex unreachable. This example shows that the greedy algorithm that chooses to merge an edge that maximizes the set of unreachable vertices given the current merges can result into a bad approximation; consider for example the case where $k \ll n$.}
  \label{fig:greedy-bad}
  \end{center}
\end{figure}

\paragraph{\bf Directed graphs.}
Another  remark is that all of our constructions can be converted to DAGs without breaking the correctness of the reductions. Figure~\ref{fig:min-tree-new} shows how to add directions to the gadget used to prove \NP-hardness for \minreach in paths under delaying. In the majority of our constructions there is a natural ``flow'' so the directions can be added trivially. For example the directions of the edges to all  tree-constructions will should be towards the leaves of the tree and the unique vertex with no incoming edges is the vertex that belongs to $S$. The only non-trivial case is the construction depicted in Figure~\ref{fig:max-reach-path}. For this construction we show in Figure~\ref{fig:directed-one} how to add directions to the edges without introducing any directed cycles. The correctness of the reduction is identical to the undirected case and for this reason it is omitted. So, all of hardness results hold in very constrained classes of directed graphs. However, real-life transportation networks are rarely acyclic and they may contain temporal cyclic paths. Hence, the following question can capture the scenario described above.

\begin{quote}
    {\em What is the complexity of \maxreach when there are at least $k$ temporal paths between any two vertices?}
\end{quote}

\begin{figure}
\begin{center}
  \includegraphics[scale=0.30]{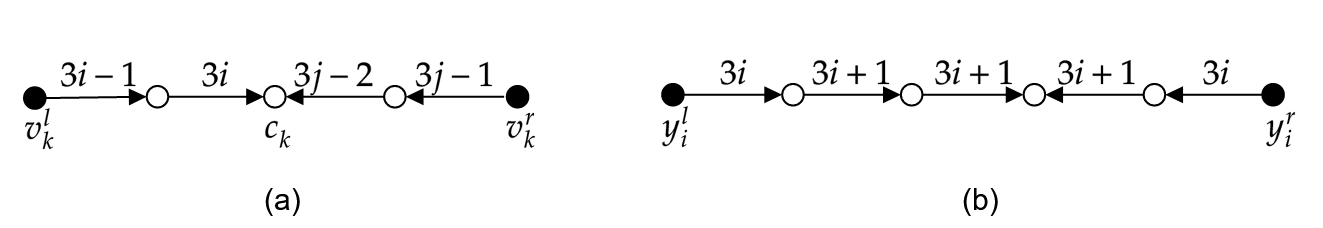}
  \caption{The directed gadgets for \maxreach on paths. Subfigure (a): c-path for the
  clause $(\bar{x}_i,x_j)$. Subfigure (b): v-path for variable $x_i$. }
  \label{fig:directed-one}
  \end{center}
\end{figure}

Another aspect of our hardness reductions is that the number of sources is a constant fraction of the vertices of the graph, if not just one. A natural question to ask then is what happens when almost all the vertices are sources, i.e. we have only a constant number of {\em non} sources. Our reductions cannot handle these cases and we conjecture that these problems can be solved efficiently. In fact, we conjecture that there exists an FPT algorithm parameterized by the number of non sources.

\paragraph{\bf Unit-disk graphs.}
Less obvious observation is that all of our \NP-hardness results for merging hold even for unit disk graphs. This observation is obvious for paths, but it is not completely  obvious for the rest of the graph classes we have constructed, since there is no direct drawing of these {\em specific} graphs via unit disk graphs. However, we can get \NP-hardness with some minor modifications. Figure~\ref{fig:unit-disks} shows how to modify the construction from Figure~\ref{fig:min-one} in order to get a unit-disk graph. More specifically, we slightly modify the the graph by adding the $q'$-vertices and updating the labels as it is shown in Figure~\ref{fig:unit-disks}(b). Following exactly the same lines of proof as in Section~\ref{sec:min-one}, we can prove that \minreach remains \NP-hard in the created instance. Then, Figure~\ref{fig:unit-disks}(c) shows how to draw this graph as a unit-disk. Using similar tricks we can extend the rest of our \NP-hardness results for unit-disk graphs for \maxreach; these constructions are easy to derive and they are left as exercises to the reader.

\begin{figure}
\begin{center}
  \includegraphics[scale=0.30]{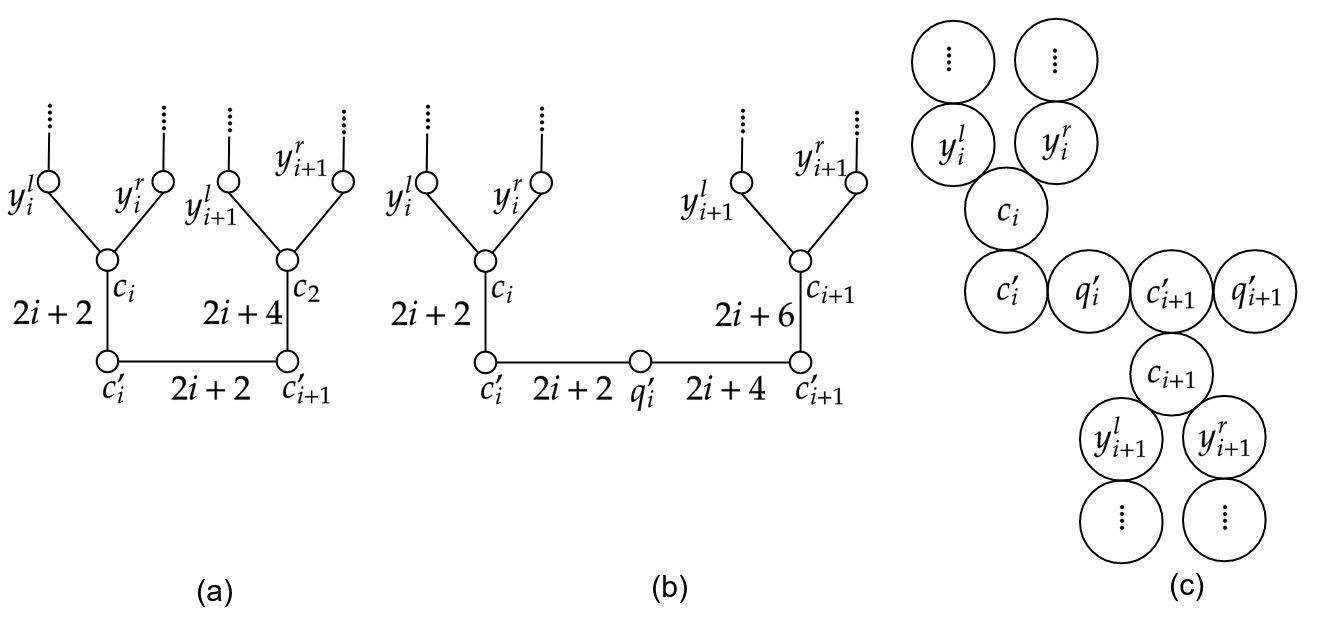}
  \caption{Subfigure (a): The gadget from \minreach on trees. Subfigure (b): The addition of the $q'$-vertices and the modification of the labels. Subfigure (c): Drawing of the gadget as a unit-disk graph. }
  \label{fig:unit-disks}
  \end{center}
\end{figure}

\paragraph{\bf Delaying and maximization objectives.}
For delaying we studied only the minimization problems. We think that maximization problems under delaying operations is another important problem to study. In this setting there is a plethora of problems to ask, all of them having immediate applications to real life problems. These questions can ask only for reachability objectives, or taking the time needed to reach a vertex into account.  When time is taken into account then delaying operations can help us to reach {\em faster} some vertices; see Figure~\ref{fig:delaying-fast}. However, some times this comes at a cost of making some other vertices unreachable ( Fig.~\ref{fig:delaying-fast}). We choose to state a more fundamental question related maximum reachability that it does not have an obvious answer. Given a temporally disconnected temporal graph, can we make it connected using delaying operations?

\begin{figure}
\begin{center}
  \includegraphics[scale=0.30]{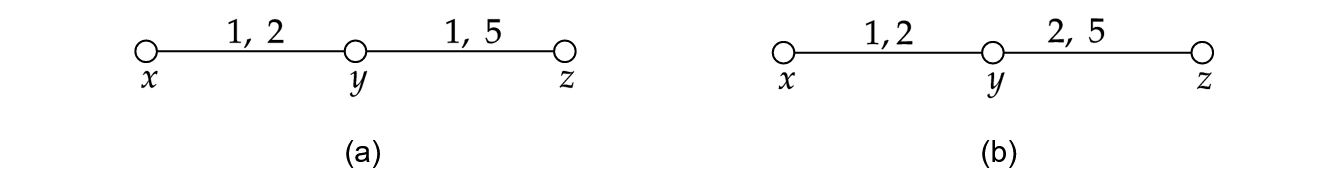}
  \caption{Subfigure (a): The initial graph. $x$ reaches $z$ at time step 5. Subfigure (b): The delaying of label 1 to 2 for edge $yz$ allowed $x$ to reach $z$ at time step 2. Observe though, this delay made $x$ ureachable from $z$.}
  \label{fig:delaying-fast}
  \end{center}
\end{figure}

\paragraph{\bf Network re-optimization.}
As we have seen our work creates many interesting questions for future research both related to reachability questions and to temporal graphs in general. Reachability objectives on temporal graphs can be studied under other notions of temporal paths, restless temporal paths is such an example~\cite{casteigts2019computational,thejaswi2020restless}. We view the our paper as a first step towards a new conceptual research direction in temporal graphs; network re-optimization. 

\begin{quote}
{\em Given a temporal network with an {\em existing} solution for a problem, can we utilize the current infrastructure in a better way and improve the solution without significantly changing the network?}
\end{quote}
\noindent
 We believe that the answer is positive and that it deserves to be further studied.

\section*{Acknowledgements}
The authors would like to thank Valentin Bura and Vladimir Gusev for several interesting 
discussions on reachability sets of temporal graphs at the beginning of the project. In addition, the authors would like to thank the anonymous reviewer whose insightful comments significantly helped us to improve the write-up and the consistency of the paper.


\bibliographystyle{abbrv}
\bibliography{references}

\end{document}